\documentclass[journal, draft,onecolumn]{IEEEtran}
\usepackage[utf8]{inputenc} 
\usepackage[T1]{fontenc}
\usepackage{url}
\usepackage{ifthen}
\usepackage{cite}
\usepackage[cmex10]{amsmath} % Use the [cmex10] option to ensure complicance
\usepackage{todonotes}
\usepackage{amssymb}
\usepackage{verbatim}
\usepackage{bm}
\usepackage{color,graphicx,xcolor}
\usepackage{mdwtab}
\usepackage{subfigure}
\usepackage{mathtools,tikz}
\usepackage{hhline}
\usepackage{multirow}
\usepackage{pdfpages}
\usepackage{enumitem}
\usepackage{balance}

\usepackage{cite}
\usepackage{amsfonts,xspace}
\usepackage{algorithmic}
\usepackage{graphicx}
\usepackage{textcomp}

\usepackage{algorithmic}
\usepackage{textcomp}
\usepackage{amsthm}
\usepackage{pgfplots}
\pgfplotsset{compat=1.15}
\usepackage{mathrsfs}
\usetikzlibrary{arrows}
\usepackage{adjustbox}
\usepackage{booktabs}
\usepackage{authblk}

\newcommand{\sX}{\mathcal{X}}
\newcommand{\sP}{\mathcal{P}}
\newcommand{\sQ}{\mathcal{Q}}

\newcommand{\pZeroStar}{p_1^*}
\newcommand{\qZeroStar}{q_1^*}
\newcommand{\pOneStar}{p_0^*}
\newcommand{\qOneStar}{q_0^*}
\newcommand{\pHStar}{p_{\textup{H}}^*}
\newcommand{\qHStar}{q_{\textup{H}}^*}

\newcommand{\tauN}{\tau_n}
\newcommand{\tauZeroN}{\tau_{0,n}}

\newcommand{\expHoeffding}{\mathcal{E}_{r}}

\newcommand{\expec}{\mathbb{E}}
\newcommand{\indicator}{\mathrm{1}}
\newcommand{\bern}{\textup{Bern}}

\newcommand{\inb}[1]{\left\{#1\right\}}
\newcommand{\inp}[1]{\left(#1\right)}
\newcommand{\insq}[1]{\left[#1\right]}

\usepackage{pgffor}
\foreach \x in {a,...,z}{%
	\expandafter\xdef\csname vec\x \endcsname{\noexpand\ensuremath{\noexpand\bm{\x}}}
}

\foreach \x in {A,...,Z}{%
	\expandafter\xdef\csname vec\x \endcsname{\noexpand\ensuremath{\noexpand\bm{\x}}}
}

\foreach \x in {A,...,Z}{%
	\expandafter\xdef\csname c\x \endcsname{\noexpand\ensuremath{\noexpand\mathcal{\x}}}
}

\foreach \x in {A,...,Z}{%
	\expandafter\xdef\csname bb\x \endcsname{\noexpand\ensuremath{\noexpand\mathbb{\x}}}
}

\DeclareMathOperator*{\argmin}{arg\,min}

\newtheorem{theorem}{Theorem}

\newtheorem{lemma}{Lemma}

\begin{document}
	\title{Sequential Adversarial Hypothesis Testing}
	
\author[1]{Eeshan Modak}
\author[2]{Mayank Bakshi}
\author[3]{Bikash Kumar Dey}
\author[1]{Vinod M. Prabhakaran}
\affil[1]{Tata Institute of Fundamental Research, Mumbai, India}
\affil[2]{Arizona State University, Tempe, AZ, USA}
\affil[3]{Indian Institute of Technology Bombay, Mumbai, India}

\maketitle

%\bibliographystyle{ieeetr}

%\bibliography{refs}

\begin{abstract}
	We study the adversarial binary hypothesis testing problem \cite{brandao2020adversarial} in the sequential setting. Associated with each hypothesis is a closed, convex set of distributions. Given the hypothesis, each observation is generated according to a distribution chosen (from the set associated with the hypothesis) by an adversary who has access to past observations. In the sequential setting, the number of observations the detector uses to arrive at a decision is variable; this extra freedom improves the asymptotic performance of the test. We characterize the closure of the set of achievable pairs of error exponents. We also study the problem under constraints on the number of observations used and the probability of error incurred.
\end{abstract}

\section{Introduction}
In binary hypothesis testing, the goal is to distinguish between two distributions (sources) \cite{chernoff1952measure},\cite{hoeffding1965asymptotically}, say $p$ (hypothesis $H_0$) or $q$ (hypothesis $H_1$). In fixed length hypothesis testing, the tester has access to
$n$ independent and identically distributed (i.i.d.) samples from the true distribution. The Neyman-Pearson lemma \cite{neyman1933problem} states that the likelihood ratio test obtains the optimal tradeoff between false alarm (type-I error) and missed detection (type-II error). If we restrict the type-I error exponent to be at least $r>0$, the maximum type-II error exponent is given by $\min\{D(s\|q): D(s\|p) \le r\}$ \cite{blahut1974hypothesis}. The Chernoff-Stein lemma \cite[Theorem 11.8.3]{cover1999elements} states that the optimal type-II error exponent when the type-I error probability is less than some $\epsilon \in (0,1)$ is given by the relative entropy $D(p\|q)$ between the two distributions.

In binary composite hypothesis testing, each hypothesis is associated with a set of distributions. Given a hypothesis, the observations are i.i.d. according to a \emph{fixed} distribution from the associated set. Zeitouni, Ziv and Merhav \cite{zeitouni1992generalized} analyse the conditions under which the generalized likelihood ratio test (GLRT) is optimal. In a variation of this problem, under a given hypothesis, observations could be arbitrarily distributed according to any of the distributions in the associated set. The choice of the distributions can be viewed to be adversarial. Fangwei and Shi \cite{fangwei1996hypothesis} study the problem where the adversary is non-adaptive but stochastic. Subsequently, Brand\~{a}o, Harrow, Lee, and Peres~\cite{brandao2020adversarial} consider the case with an adaptive adversary which has feedback of the past observations and can use this information to select the distribution of the next sample. They showed that the Chernoff-Stein exponent in the adversarial setting is $\min \limits_{p \in \mathcal{P}, q \in \mathcal{Q}} D(p\|q)$.
\\

In sequential binary hypothesis testing, the detector stops and outputs a decision only when it develops enough confidence. The length of the test is a stopping time subject to some constraint. A commonly studied constraint is the expectation constraint which says that in expectation (under $H_0$ and $H_1$) the stopping time should be less than or equal to $n$. Wald and Wolfowitz \cite{wald1948optimum} show in the sequential setting that
the pair of exponents $(D(q\|p),D(p\|q))$ can be simultaneously achieved by the sequential probability ratio test (SPRT) with appropriate thresholds. Lalitha and Javidi \cite{lalitha2016reliability} study a hybrid setting where one could draw more than $n$ samples with exponentially small probability, $e^{-n\gamma}$. Here, the parameter $\gamma$ interpolates between the fixed length ($\gamma \rightarrow \infty$) and the sequential setting ($\gamma = 0$). Lai \cite{lai2002asymptotic}, and Li, Liu and Ying \cite{li2014generalized} consider the sequential testing problem of composite hypotheses and show the asympotic optimality of generalized SPRT in certain cases. Li and Tan \cite{li2020second} analyse the second order asymptotics of sequential hypothesis testing. Pan, Li and Tan \cite{pan2022asymptotics} consider the sequential testing problem where one of the hypotheses is given by a single distribution and the other is given by a parameterized family of distributions. They study the problem under the constraint that the test will require $n$ or fewer samples with high probability and characterize first and second order exponents under certain conditions.

\subsection{Our work}
In this work, we study the adversarial hypothesis testing problem \cite{brandao2020adversarial} under the sequential setting. Associated with each hypothesis is a closed, convex set of distributions. All the distributions have the same support. Given the hypothesis, at each time instance, the adversary chooses a distribution from which the next observation will be sampled. The adversary has feedback of the past observations and can make an adaptive choice. The test is given by a stopping rule subject to certain constraints on the number of observations and a corresponding decision rule. We first show that we can strictly improve the set of achievable error exponents and in particular simultaneously achieve the pair $\left(\min \limits_{p \in \mathcal{P}, q \in \mathcal{Q}} D(q\|p),\min \limits_{p \in \mathcal{P}, q \in \mathcal{Q}} D(p\|q)\right)$. Furthermore, a modified SPRT (with two different ratios) can achieve this point. We also study the problem under a probability constraint on the stopping time and show that the set of achievable error exponents is a rectangular region with the corner point $\left(\min \limits_{p \in \mathcal{P}, q \in \mathcal{Q}} D(q\|p),\min \limits_{p \in \mathcal{P}, q \in \mathcal{Q}} D(p\|q)\right)$. Lastly, we study the problem under a constraint on the probability of error and show that the set of achievable error exponents is a rectangular region with the corner point $\left( \min \limits_{p \in \mathcal{P}, q \in \mathcal{Q}} D(p\|q),\min \limits_{p \in \mathcal{P}, q \in \mathcal{Q}} D(q\|p) \right)$.

\subsection{Notation}
We use the notation $X^t$ to denote the set of random variables $X_1,\dots,X_t$. Likewise, $x^t$ refers to the sequence $x_1,\ldots,x_t$. $\bern(p)$ is the Bernoulli distribution with parameter $p$ and $h(p):=-p\log(p) - (1-p)\log(1-p)$ denotes the binary entropy function. $\mathrm{1}_{E}$ is the indicator random variable for the event $E$. We use $a \land b$ to denote the minimum of $a$ and $b$.
\section{Problem Setup} \label{sec:problem_setup}
Let $\sX$ be a finite alphabet. Consider two sets of distributions $\sP,\sQ \subseteq \mathbb{R}^{\sX}$ with the following properties: \\
(A1) $\sP$ and $\sQ$ are closed and convex. We assume that $\sP \cap \sQ = \emptyset$ to make the problem non-trivial. \\
(A2) All pairs of distributions $(p,q)$, where $p \in \sP$ and $q \in \sQ$, are mutually absolutely continuous. \\
Consider the following adversarial hypothesis testing problem with
\begin{align*}
&H_0 : \sP \\
&H_1 : \sQ.
\end{align*}
The adaptive adversary is specified by the functions $\hat{p}_t: \sX^{t-1} \rightarrow \sP$ and $\hat{q}_t: \sX^{t-1} \rightarrow \sQ$ for $t \in \{1,2,\ldots \}$. For $t=1$, the function just outputs an arbitrary distribution in $\sP$ (or $\sQ$). The adversary chooses the set of functions $(\hat{p}_t)_{t \in \mathbb{N}}$ (respectively $(\hat{q}_t)_{t \in \mathbb{N}}$) under $H_0$ (resp. $H_1$). The law of $X^t$ under $H_0$ is given by $\mathbb{P}(x^t)=\prod_{i=1}^{t}p(x_i|x^{i-1})$ where $p(.|x^{i-1})=\hat{p}_i(x^{i-1})$ and under $H_1$ by $\mathbb{Q}(x^t)=\prod_{i=1}^{t}q(x_i|x^{i-1})$ where $q(.|x^{i-1})=\hat{q}_i(x^{i-1})$. A test $\phi$ is defined by the pair $(\tau, Z)$, where $\tau$ is a stopping time adapted to the filtration $\mathcal{F}_0 \subseteq \mathcal{F}_1 \cdots \subseteq \mathcal{F}_t \cdots \subseteq \mathcal{F}$, $\mathcal{F}_t :=\ \sigma\{X_1,\ldots,X_t\}$ and $Z : \sX^{\tau} \rightarrow \{0,1\}$ is a $\mathcal{F}_{\tau}$-measurable function that specifies the decision rule. Here, $\sX^{\tau}$ denotes the set of all stopped sequences. For a particular choice of adversary strategy $(\hat{p}_t,\hat{q}_t)_{t \in \mathbb{N}}$ and for a given stopped sequence $x^{\tau}$, $\mathbb{P}$ and $\mathbb{Q}$ are given by 
\begin{equation}
\mathbb{P}(x^{\tau}) = \prod_{i=1}^{\tau} \hat{p}_i(x^{i-1})(x_i)
\end{equation}
\begin{equation}
\mathbb{Q}(x^{\tau}) = \prod_{i=1}^{\tau} \hat{q}_i(x^{i-1})(x_i).
\end{equation}
Note that $\hat{p}_i(x^{i-1})$ and $\hat{q}_i(x^{i-1})$ are distributions in $\mathcal{P}$ and $\mathcal{Q}$ respectively. For conciseness, let $\mathcal{A}$ denote the adversary strategy $(\hat{p}_t,\hat{q}_t)_{t \in \mathbb{N}}$. For a particular adversary strategy, the type-I (false alarm) and type-II (missed detection) error are defined as follows,
\begin{align*}
\pi_{1|0}(\phi,\mathcal{A}) &:= \mathbb{P}(Z = 1) \\
\pi_{0|1}(\phi,\mathcal{A}) &:= \mathbb{Q}(Z = 0).
\end{align*}
A pair of exponents $(E_0,E_1)$ is said to be \emph{achievable}, if there exists a sequence of tests $(\phi_n=(\tauN,Z_n))_{n \in \mathbb{N}}$ such that
\begin{align*}
E_0 &\le \liminf_{n \rightarrow \infty} -\frac{\log \sup_{\mathcal{A}_n} \pi_{1|0}(\phi_n, \mathcal{A}_n)}{\sup_{\mathcal{A}_n} \mathrm{E}_{\mathbb{P}}[\tau_n]} \\
E_1 &\le \liminf_{n \rightarrow \infty} -\frac{\log \sup_{\mathcal{A}_n} \pi_{0|1}(\phi_n, \mathcal{A}_n)}{\sup_{\mathcal{A}_n} \mathrm{E}_{\mathbb{Q}}[\tau_n]},
\end{align*}
and $\sup \limits_{\mathcal{A}_n} \mathrm{E}_{\mathbb{P}}[\tau_n] < \infty$, $\sup \limits_{\mathcal{A}_n} \mathrm{E}_{\mathbb{Q}}[\tau_n] < \infty$, and $\sup \limits_{\mathcal{A}_n} \mathrm{E}_{\mathbb{P}}[\tau_n]$, $\sup \limits_{\mathcal{A}_n} \mathrm{E}_{\mathbb{Q}}[\tau_n]$ are unbounded as $n \rightarrow \infty$.
Define $\mathcal{E}(\mathcal{P},\mathcal{Q})$ to be the set of all achievable $(E_0,E_1)$ pairs. We want to characterize this region.
%In words, the expected stopping time for the test $\phi_n$ must not exceed $n$ for all possible adversary strategies $\mathcal{A}_n$ and asympotically the tests achieve the error exponents $(E_0,E_1)$.
\section{Results} \label{sec:results}
Let $(p_0^*,q_0^*)$ be the closest pair between $\sP$ and $\sQ$ in the forward Kullback-Leibler (KL) divergence, i.e. 
\begin{equation} \label{eqn:opt_pair1}
(p_0^*,q_0^*)  := \argmin \limits_{p \in \mathcal{P}, q \in \mathcal{Q}} D(p\|q).
\end{equation}
Let $(p_1^*,q_1^*)$ be the closest pair in the reverse KL divergence sense, i.e. 
\begin{equation} \label{eqn:opt_pair0}
(p_1^*,q_1^*)  := \argmin \limits_{p \in \mathcal{P}, q \in \mathcal{Q}} D(q\|p)
\end{equation}
Since $\mathcal{P}$ and $\mathcal{Q}$ are closed, convex sets and KL divergence is lower semi-continuous in its arguments, such minima exist. Since we assume that all distributions in $\sP$ and $\sQ$ are mutually absolutely continuous, let $\max \limits_{x \in \sX}|\log \frac{\pZeroStar(x)}{\qZeroStar(x)}| = c_0$ and $\max \limits_{x \in \sX}|\log \frac{\pOneStar(x)}{\qOneStar(x)}| = c_1$.

\par
We can extend the ideas in \cite{brandao2020adversarial} to get a complete (optimal) tradeoff curve for the fixed length adversarial hypothesis testing problem (refer Appendix~\ref{app:tradeoff}). In the sequential setting, we can get a strictly better tradeoff as shown in Figure~\ref{fig:improved_tradeoff}.
\begin{theorem} \label{thrm:bemv1}
	For sets $\sP,\sQ \subseteq \mathbb{R}^{\sX}$ which satisfy the properties A1 and A2, the set of achievable pairs of exponents is
	\begin{equation}
	\mathcal{E}(\mathcal{P},\mathcal{Q}) = \inb{ (E_0,E_1):E_0 E_1 \le D(\qZeroStar\|\pZeroStar)D(\pOneStar\|\qOneStar)}.
	\end{equation}
	Furthermore, all points on the curve $(E_0 E_1 = D(\qZeroStar\|\pZeroStar)D(\pOneStar\|\qOneStar))$ are achievable.
\end{theorem}

\par
We also study two related problems. In Section~\ref{app:prob_constraint}, we characterize the error exponents under a probability constraint on the stopping time. In Section~\ref{app:dual_problem}, we characterize the error exponents under a constraint on the probability of error.

\begin{figure}
    \tikzset{every picture/.style={line width=0.75pt}} %set default line width to 0.75pt        
\centering
\begin{tikzpicture}[x=0.75pt,y=0.75pt,yscale=-1,xscale=1]
%uncomment if require: \path (0,384); %set diagram left start at 0, and has height of 384

%Shape: Axis 2D [id:dp8815546877731582] 
\draw  (140.5,323.4) -- (509.5,323.4)(177.4,48) -- (177.4,354) (502.5,318.4) -- (509.5,323.4) -- (502.5,328.4) (172.4,55) -- (177.4,48) -- (182.4,55)  ;
%Curve Lines [id:da28710599551606075] 
\draw    (177.5,176) .. controls (197,260) and (302.5,324) .. (364.5,322) ;
%Curve Lines [id:da41852705100373266] 
\draw [line width=2.25]    (296.5,75) .. controls (316,159) and (421.5,223) .. (483.5,221) ;
%Straight Lines [id:da7707732658243266] 
\draw  [dash pattern={on 4.5pt off 4.5pt}]  (177.5,176) -- (363.5,174) ;
%Straight Lines [id:da6042086790247282] 
\draw  [dash pattern={on 4.5pt off 4.5pt}]  (364.5,174) -- (364.5,322) ;

% Text Node
\draw (465,348.4) node [anchor=north west][inner sep=0.75pt]  [font=\large]  {$E_{1}$};
% Text Node
\draw (142,75.4) node [anchor=north west][inner sep=0.75pt]  [font=\large]  {$E_{0}$};
% Text Node
\draw (83,163.4) node [anchor=north west][inner sep=0.75pt]  [font=\large]  {$D\left( q_{1}^{*} \| p_{1}^{*}\right)$};
% Text Node
\draw (324,335.4) node [anchor=north west][inner sep=0.75pt]  [font=\large]  {$D\left( p_{0}^{*} \| q_{0}^{*}\right)$};

\end{tikzpicture}
    \caption{The lighter curve shows the tradeoff between error exponents in the fixed length setting. The darker curve shows the tradeoff in the sequential setting. As we can see the pair $(E_1=D(\pZeroStar\|\qZeroStar),E_0=D(\qOneStar\|\pOneStar))$ can be achieved in the sequential case. }
    \label{fig:improved_tradeoff}
\end{figure}
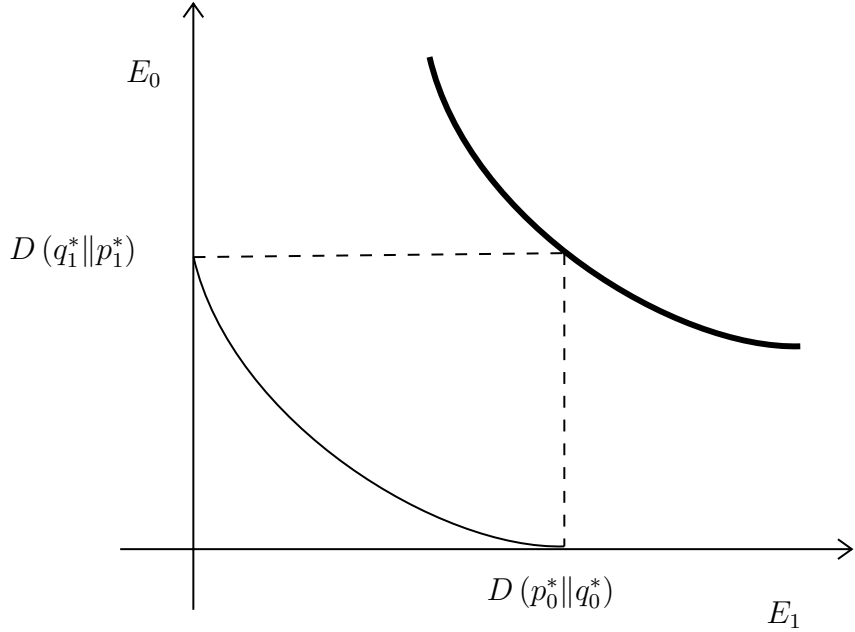

\section{Key Lemmas}
We first state two lemmas used in our proof.
\begin{lemma} \label{lemma:pyth_thrm}
	For any $p \in \mathcal{P}$ and $X$ distributed according to $p$,
	\begin{equation} \label{eqn:pyth_thrm1}
	\expec \insq{\log \frac{\pOneStar(X)}{\qOneStar(X)}} \ge D(\pOneStar\|\qOneStar),
	\end{equation}
	with equality if $p=\pOneStar$.
	For any $q \in \mathcal{Q}$ and $X$ distributed according to $q$,
	\begin{equation} \label{eqn:pyth_thrm0}
	\expec \insq{\log \frac{\qZeroStar(X)}{\pZeroStar(X)}} \ge D(\qZeroStar\|\pZeroStar).,
	\end{equation}
	with equality if $q=\qZeroStar$.
\end{lemma}
\begin{proof}
	By \eqref{eqn:opt_pair1}, we see that $\pOneStar$ is (one of) the closest distribution in $\mathcal{P}$ to $\qOneStar$. Using the Pythagorean inequality for KL divergence \cite[Theorem 11.6.1]{cover1999elements}, we get
	\begin{equation*}
	D(p\|\qOneStar) \ge D(p\|\pOneStar) + D(\pOneStar\|\qOneStar), \hspace{10pt} p \in \mathcal{P}  
	\end{equation*}
	\begin{align*}
	\textup{i.e.}, \hspace{10pt} \sum_{x} p(x) \log \frac{p(x)}{\qOneStar(x)} \ge &\sum_{x} p(x) \log \frac{p(x)}{\pOneStar(x)} \\
	&+ \sum_{x} \pOneStar(x) \log \frac{\pOneStar(x)}{\qOneStar(x)}.
	\end{align*}
	By rearranging terms, we get \eqref{eqn:pyth_thrm1}. If $p=\pOneStar$, then the Pythagorean inequality becomes an equality. The proof of \eqref{eqn:pyth_thrm0} is similar.
\end{proof}

\begin{lemma}[Lemma 5, \cite{brandao2020adversarial}]
	For any $q \in \mathcal{Q}$ and $X \sim q$,
	\begin{equation} \label{eqn:likelihood_ineq1}
	\expec \insq{\frac{\pOneStar(X)}{\qOneStar(X)}} \le 1. 
	\end{equation}
	For any $p \in \mathcal{P}$ and $X \sim p$,
	\begin{equation} \label{eqn:likelihood_ineq2}
	\expec \insq{\frac{\qZeroStar(X)}{\pZeroStar(X)}} \le 1. 
	\end{equation}
\end{lemma}

For a random process $\{X_t\}_t$, $S_{0,t}$ and $S_{1,t}$ be the log likelihood functions of the first $t$ observations with respect to $(\pOneStar,\qOneStar)$ and $(\qZeroStar,\pZeroStar)$ respectively.
\begin{equation} \label{eqn:sum_k}
S_{0,t} := \sum_{i=1}^{t} \log \frac{\pOneStar(X_i)}{\qOneStar(X_i)}, \hspace{10pt}  S_{1,t} := \sum_{i=1}^{t} \log \frac{\qZeroStar(X_i)}{\pZeroStar(X_i)}.
\end{equation}
Define $\{M_{0,t}\}_{t \in \mathbb{N}}$ and $\{M_{1,t}\}_{t \in \mathbb{N}}$ to be the following processes
\begin{equation*}
M_{0,t} := S_{0,t} - tD(\pOneStar\|\qOneStar), \hspace{10pt} M_{1,t} := S_{1,t} - tD(\qZeroStar\|\pZeroStar).
\end{equation*}
Also, let $M_{0,0}:=0$, $M_{1,0}:=0$. \\
\textbf{Claim 1}: Under $H_0$ (i.e. when the observations are distributed according to $\mathbb{P}$), $\{M_{0,t}\}_t$ is a submartingale. If the observations are drawn i.i.d. according to $\pOneStar$, then $\{M_{0,t}\}_t$ is a martingale. Likewise for $\{M_{1,t}\}_t$ under $H_1$.\\
\textit{Proof}: We will prove the statement for $\{M_{0,t}\}_t$, the statement for $\{M_{1,t}\}_t$ follows by a symmetric argument.  
\begin{align*}
\mathbb{E}[M_{0,t}|X_{1}^{t-1}] &= \mathbb{E}[S_{0,t}|X_{1}^{t-1}] - tD(\pOneStar\|\qOneStar) \\
&= S_{0,t-1} \! + \! \mathbb{E}\insq{\log \frac{\pOneStar(X_t)}{\qOneStar(X_t)}\bigg|X_{1}^{t-1}} \! -  \! tD(\pOneStar\|\qOneStar) \\
&\ge S_{0,t-1} - (t-1)D(\pOneStar\|\qOneStar) \\
&= M_{0,t-1}, 
\end{align*}
where the inequality follows from lemma~\ref{lemma:pyth_thrm}. Observe that the penultimate step is an equality if the observations are drawn i.i.d. according to $\pOneStar$. \\
Let $\{\tilde{M}_{0,t}\}_{t \in \mathbb{N}}$ and $\{\tilde{M}_{1,t}\}_{t \in \mathbb{N}}$ be the stopped processes.
\begin{equation*}
\tilde{M}_{0,t} :=  M_{0,t \land \tau} \hspace{10pt} \tilde{M}_{1,t} :=  M_{1,t \land \tau}.
\end{equation*}
Here, $t \land \tau = \min(t,\tau)$. \\
\textbf{Claim 2:} Under $H_0$, $\{\tilde{M}_{0,t}\}_t$ is a submartingale. If the observations are drawn i.i.d. according to $\pOneStar$, then $\{\tilde{M}_{0,t}\}_t$ is a martingale. Likewise for $\{\tilde{M}_{1,t}\}_t$ under $H_1$.\\
\textit{Proof:} The claim follows from $M_{0,t}$ (resp. $M_{1,t}$) being a submartingale under $H_0$ (resp. $H_1$) by Claim 1 and the fact that a stopped process is a submartingale if the original process is a submartingale \cite[Theorem 3.10]{cinlar2011probability}.
\par
Consider the following stopping times.
\begin{align}
\tau_0 := \inf \{t \in \mathbb{N}: S_{0,t} \geq \theta_0\} \label{eqn:tau0} \\
\tau_1 := \inf \{t \in \mathbb{N}: S_{1,t} \geq \theta_1\} \label{eqn:tau1}
\end{align}
where $\theta_0,\theta_1>0$.
\begin{lemma} \label{lemma:exp_tau_ub}
	For any stopping rule, in particular for a stopping rule defined in \eqref{eqn:tau0}, we have
 \begin{equation} \label{eqn:ub_tau0}
	\mathbb{E}_{\mathbb{P}}[\tau_0] \le \frac{\expec_{\mathbb{P}} \insq{S_{0,\tau_0}}}{D(\pOneStar\|\qOneStar)}.
	%& \le \frac{n(D(\pOneStar\|\qOneStar)-\delta) + c_1}{D(\pOneStar\|\qOneStar)}	
	\end{equation}
	Likewise, for any stopping rule, in particular for a stopping rule defined in \eqref{eqn:tau1}, we have
	\begin{equation} \label{eqn:ub_tau1}
	\mathbb{E}_{\mathbb{Q}}[\tau_1] \le \frac{\expec_{\mathbb{Q}} \insq{S_{1,\tau_1}}}{D(\qZeroStar\|\pZeroStar)}.
	%& \le \frac{n(D(\pOneStar\|\qOneStar)-\delta) + c_1}{D(\pOneStar\|\qOneStar)}	
	\end{equation}
\end{lemma}
\begin{proof}
	We will prove \eqref{eqn:ub_tau0}. The proof of \eqref{eqn:ub_tau1} follows by a symmetric argument. Recall that by Claim 2, $\{\tilde{M}_{0,t}\}_t$ is a submartingale under $H_0$. Thus, we have 
	\begin{equation*} 
	\mathbb{E}_{\mathbb{P}}[\tilde{M}_{0,t}] \ge \mathbb{E}_{\mathbb{P}}[\tilde{M}_{0,0}] = 0.
	\end{equation*}
	Plugging in the definition of $\tilde{M}_{0,t}$,
	\begin{equation} \label{eqn:submart}
	\mathbb{E}_{\mathbb{P}}[S_{0,t \land \tau_0}] \ge \mathbb{E}_{\mathbb{P}}[t \land \tau_0]D(\pOneStar\|\qOneStar).
	\end{equation}
	To obtain \eqref{eqn:ub_tau0} from this, we will make the following claim which is proved later. \\
	\textbf{Claim 3}: $\mathbb{E}_{\mathbb{P}}[\tau_0] < \infty$, $\mathbb{E}_{\mathbb{Q}}[\tau_1] < \infty$. \\
	Then, $\tau_0$ is absolutely integrable. Since $t \land \tau_0 \le \tau_0$, by the dominated convergence theorem,
	\begin{equation} \label{eqn:rhs_dct}
	\lim_{t \rightarrow \infty} \mathbb{E}_{\mathbb{P}}[t \land \tau_0] = \mathbb{E}_{\mathbb{P}}[\lim_{t \rightarrow \infty} t \land \tau_0] = \expec_{\mathbb{P}}[\tau_0].
	\end{equation}
	Since $|S_{0,t \land \tau_0}| \le (t \land \tau_0)c_0 \le \tau_{0}c_{0}$, by absolute integrability of $\tau_1$ and the dominated convergence theorem,
	\begin{equation} \label{eqn:lhs_dct}
	\lim_{t \rightarrow \infty} \mathbb{E}_{\mathbb{P}}[S_{0,t \land \tau_0}] = \expec_{\mathbb{P}}[\lim_{t \rightarrow \infty} S_{0,t \land \tau_0}] = \expec_{\mathbb{P}}[S_{0,\tau_0}].
	\end{equation}
	Observe that \eqref{eqn:ub_tau0} follows from \eqref{eqn:submart}, \eqref{eqn:rhs_dct}, \eqref{eqn:lhs_dct}. 
	It only remains to prove Claim 3. \\
	\textit{Proof of Claim 3:}
	We show $\expec_{\mathbb{P}}[\tau_{0}] < \infty$; $\expec_{\mathbb{Q}}[\tau_{1}] < \infty$ follows similarly.
	\begin{align*}
	\mathbb{E}_{\mathbb{P}}[\tau_{0}] &= \sum_{i=1}^{\infty} \mathbb{P}[\tau_{0} \ge i] \\
	&= \sum_{i=1}^{N} \mathbb{P}[\tau_{0} \ge i] + \sum_{i=N+1}^{\infty} \mathbb{P}[\tau_{0} \ge i]
	\end{align*}
	Here, $N$ is specified below. The first sum is clearly finite. We need to show that the tail sum is bounded. Towards that end, we examine each quantity in the tail sum. Define the events $E_j := \{S_{0,j} < \theta_0\}$, $j \in \mathbb{N}$.
	\begin{align*}
	&\mathbb{P}[\tau_{0} \ge i] \\
	&= \mathbb{P} \insq{\cap_{j=1}^{i-1} E_j} \\
	&\le \mathbb{P} \insq{E_{i-1}} \\
	&= \mathbb{P} \insq{S_{0,i-1} < \theta_0} \\
	&= \mathbb{P} \insq{M_{0,i-1} < \theta_0 -(i-1)D(\pOneStar\|\qOneStar)}
	\end{align*}
	Let $N$ be such that $N \ge \frac{\theta_0}{D(\pOneStar\|\qOneStar)}$. Then, for $i \ge N+1$, $\theta_0 -(i-1)D(\pOneStar\|\qOneStar) < 0$. By Claim 1, $\{M_{0,t}\}_t$ is a submartingale. Since $\max \limits_{x \in \sX}|\log \frac{\pOneStar(x)}{\qOneStar(x)}| = c_0$, we have $|M_{0,t}-M_{0,t-1}| \le 2c_0$. We apply Azuma-Hoeffding inequality for submartingale sequences with bounded differences to get,
	\begin{align*}
	\mathbb{P}[\tau_{0,n} \ge i] &\le \mathbb{P} \insq{M_{0,i-1} < \theta_0 -(i-1)D(\pOneStar\|\qOneStar)} \\
	&\le \exp\inp{ \frac{-(\theta_0 -(i-1)D(\pOneStar\|\qOneStar))^2}{2\sum_{k=1}^{i-1}(2c_0)^2} } \\
	&= \exp \inp{\frac{-(i-1)^2\inp{D(\pOneStar\|\qOneStar) - \frac{\theta_0}{i-1}}^2}{8(i-1)c_0^2}} \\
	&= \exp \inp{\frac{-(i-1)\inp{D(\pOneStar\|\qOneStar) - \frac{\theta_0}{i-1}}^2}{8c_0^2}}
	\end{align*}
	
	The tail sum can now be written as follows
	\begin{align*}
	&\sum_{i=N+1}^{\infty} \mathbb{P}[\tau_{0} \ge i] \\
	&\le \sum_{i=N+1}^{\infty}  \exp  \inp{ \frac{-(i-1)\inp{D(\pOneStar\|\qOneStar) - \frac{\theta_0}{i-1}}^2}{8c_0^2}} \\
	&\le \sum_{i=N+1}^{\infty}  \exp \inp{\frac{-(i-1)\inp{D(\pOneStar\|\qOneStar) - \frac{\theta_0}{N}}^2}{8c_0^2}} \\
	&= \sum_{i=N+1}^{\infty} \exp \inp{-(i-1)c_{\pOneStar,\qOneStar,\theta_0}} \\
	&= \frac{\exp \inp{-Nc_{\pOneStar,\qOneStar,\theta_0}}}{1-\exp(-c_{\pOneStar,\qOneStar,\theta_0})} \\
	&< \infty.
	\end{align*}
	Here, $c_{\pOneStar,\qOneStar,\theta_0} > 0$ is a constant that depends on $\pOneStar,\qOneStar$ and $\theta_0$. Thus, $\expec_{\mathbb{P}}[\tau_0]$ (and similarly $\expec_{\mathbb{Q}}[\tau_1]$) are finite and Claim 3 is proved.\\
\end{proof}

\begin{lemma} \label{lemma:ut_ub}
	For a stopping rule as defined in \eqref{eqn:tau1}, we have
	\begin{equation}
	\expec_{\mathbb{P}} \insq{2^{S_{1,\tau_1}}} \le 1
	\end{equation}
	For a stopping rule as defined in \eqref{eqn:tau0}, we have
	\begin{equation}
	\expec_{\mathbb{Q}} \insq{2^{S_{0,\tau_0}}} \le 1
	\end{equation}
\end{lemma}
\begin{proof}
	Let $U_t := 2^{S_{1,t \land \tau_1}}$.
	\begin{align*}
	&\expec_{\mathbb{P}}[U_t] \\
	&= \expec_{\mathbb{P}} \insq{ U_{t}(\indicator_{\{\tau_1 \le t-1\}} + \indicator_{\{\tau_1 \ge t\}})} \\
	&\overset{(a)}{=} \expec_{\mathbb{P}} \insq{U_{t-1}\indicator_{\{\tau_1 \le t-1\}}} + \expec_{\mathbb{P}} \insq{U_{t-1}\frac{\qZeroStar(X_t)}{\pZeroStar(X_t)}\indicator_{\{\tau_1 \ge t\}}} \\
	& \! = \! \! \expec_{\mathbb{P}} \! \insq{U_{t-1}\indicator_{\{\tau_1 \le t-1\}}} \! + \! \expec \! \! \insq{\! \expec \! \insq{U_{t-1}\frac{\qZeroStar(X_t)}{\pZeroStar(X_t)}\indicator_{\{\tau_1 \ge t\}}\bigg|X^{t-1}}} \\
	&\! \overset{(b)}{=} \! \! \expec_{\mathbb{P}} \! \! \insq{U_{t-1}\indicator_{\{\tau_1 \le t-1\}}} \! \! +  \! \! \expec \! \! \insq{ U_{t-1}\indicator_{\{\tau_1 \ge t\}} \expec \! \insq{\frac{\qZeroStar(X_t)}{\pZeroStar(X_t)}\bigg|X^{t-1}}} \\
	&\overset{(c)}{\le} \expec_{\mathbb{P}} \insq{U_{t-1}\indicator_{\{\tau_1 \le t-1\}}} + \expec_{\mathbb{P}} \insq{U_{t-1}\indicator_{\{\tau_1 \ge t\}}} \\
	&= \expec_{\mathbb{P}}[U_{t-1}].
	\end{align*}
	The equality $(a)$ is because when $\tau_1 \le t-1$, $t \land \tau_1=(t-1) \land \tau_1$ and hence $U_t=U_{t-1}$. On the other hand, when $\tau_1 \ge t$, $t \land \tau_1=t$ and hence $U_t = U_{t-1}\frac{\qZeroStar(X_t)}{\pZeroStar(X_t)}$. $(b)$ follows from the fact that $U_{t-1}\indicator_{\{\tau_1 \ge t\}}$ is a function of $X_1,\ldots,X_{t-1}$. $(c)$ follows from \eqref{eqn:likelihood_ineq2}.
	Proceeding in a similar manner, we have
	\begin{equation} \label{eqn:U_t_non_increasing}
	\expec_{\mathbb{P}}[U_t] \le \expec_{\mathbb{P}}[U_{t-1}] \le \cdots \le \expec_{\mathbb{P}}[U_1] = \expec_{\mathbb{P}} \insq{\frac{\qZeroStar(X_1)}{\pZeroStar(X_1)}} \le 1.
	\end{equation}
	The final inequality follows from \eqref{eqn:likelihood_ineq2}.
	We take the limit as $t \rightarrow \infty$. The limit exists since $\expec_{\mathbb{P}}[U_t]$ is non-increasing in $t$ and lower bounded by $0$.
	\begin{equation} \label{eqn:before_dct1}
	\lim_{t \rightarrow \infty} \expec_{\mathbb{P}}[U_t] = \lim_{t \rightarrow \infty} \expec_{\mathbb{P}}\insq{2^{S_{1,t \land \tau_1}}} \le 1.
	\end{equation}
	Recall the definition $\tau_{1} := \inf \inb{t \in \mathbb{N}: S_{1,t} \geq \theta_1}$. Hence, for all $t \in \mathbb{N}$,
	\begin{align*}
	S_{1,t \land \tau_{1}} &\le S_{0,\tau_{1}} \\
	&= \sum_{i=1}^{\tau_{1}-1}\log \frac{\qZeroStar(X_i)}{\pZeroStar(X_i)} + \log \frac{\qZeroStar(X_{\tau_1})}{\pZeroStar(X_{\tau_1})} \\
	&\le \theta_1 + c_1.	
	\end{align*}
	The first inequality follows from the definition of $\tau_1$. The second inequality follows from the fact that at the penultimate step the sum $S_{1,\tau_{1}-1}$ cannot exceed the threshold and the final increment $\log \frac{\qZeroStar(X_{\tau_1})}{\pZeroStar(X_{\tau_1})}$ is upper bounded by $c_1$.
	Thus, we have 
	\begin{equation*}
	0 \le 2^{S_{1,t \land \tau_{1}}} \le 2^{\theta_1 + c_1}.
	\end{equation*}
	Hence, by \eqref{eqn:before_dct1} and by the bounded convergence theorem,	
	\begin{equation*}
	\expec_{\mathbb{P}} \insq{2^{S_{1,\tau_{1}}}} \le 1.
	\end{equation*}
\end{proof}
\section{Proof of Theorem~\ref{thrm:bemv1}} \label{sec:proofs}
 
\subsection{Proof of achievability of Theorem~\ref{thrm:bemv1}}
    Fix a point on the curve $(E_0 E_1 = D(\qZeroStar\|\pZeroStar)D(\pOneStar\|\qOneStar))$. It has the form $E_0=\frac{\alpha_1}{\alpha_0}D(\pOneStar\|\qOneStar)$, $E_1=\frac{\alpha_0}{\alpha_1}D(\qZeroStar\|\pZeroStar)$. We show a sequence of tests which achieves this point. Consider the tests $\phi_n=(\tauN,Z_n)$, $n \in \mathbb{N}$ as defined below. The stopping time is given by,
	\begin{equation} \label{eqn:tauN}
	\tau_n := \inf \{t \in \mathbb{N}: S_{0,t} \geq \alpha_0 n \hspace{5pt} \textup{or}\hspace{5pt} S_{1,t} \geq \alpha_1 n\}.
	\end{equation}
	The decision rule is defined to be
	\begin{equation} \label{eqn:decision_rule}
	Z_n(x^{\tauN}) := \left\{
	\begin{array}{ll}
	0 & \mbox{if } S_{0,\tauN} \geq \alpha_0 n, S_{1,\tauN} < \alpha_1 n\\
	1 & \mbox{if } S_{1,\tauN} \geq \alpha_1 n.
	\end{array}
	\right.
	\end{equation}
	We need to show that
    \begin{align*}
        \liminf_{n \rightarrow \infty} -\frac{\log \sup \limits_{\mathcal{A}_n} \pi_{1|0}(\phi_n, \mathcal{A}_n)}{\sup \limits_{\mathcal{A}_n} \mathrm{E}_{\mathbb{P}}[\tau_n]} &\ge \frac{\alpha_1}{\alpha_0}D(\pOneStar\|\qOneStar) \\ 
        \liminf_{n \rightarrow \infty} -\frac{\log \sup \limits_{\mathcal{A}_n} \pi_{0|1}(\phi_n, \mathcal{A}_n)}{\sup \limits_{\mathcal{A}_n} \mathrm{E}_{\mathbb{Q}}[\tau_n]} &\ge \frac{\alpha_0}{\alpha_1}D(\qZeroStar\|\pZeroStar).
    \end{align*}
	Define $\tau_{0,n}$ and $\tau_{1,n}$ to be
	\begin{align}
    \tau_{0,n} &:= \inf \inb{t \in \mathbb{N}: S_{0,t} \geq \alpha_0 n}. \label{eqn:tau0N_exp_constraint} \\
	\tau_{1,n} &:= \inf \inb{t \in \mathbb{N}: S_{1,t} \geq \alpha_1 n} \label{eqn:tau1N_exp_constraint}
	\end{align}
	We now bound the type-I error probability, $\pi_{1|0}(\phi_n,\mathcal{A}_n)$.
	Let $E \subseteq \mathcal{F}_{\tau_n}$ be the error event under $H_0$. Thus, $E = \{ S_{1,\tau_n} \ge \alpha_1 n\} = \{ \tau_{n} = \tau_{1,n}\}$.
	\begin{align*}
	\pi_{1|0}(\phi_n,\mathcal{A}_n) &= \mathbb{P} (E) \\
	&= \expec_{\mathbb{P}}[\mathrm{1}_{E}] \\
	&= \expec_{\mathbb{P}}\insq{\mathrm{1}_{E} 2^{-S_{1,\tau_{1,n}}} 2^{S_{1,\tau_{1,n}}}} \\
	%&= \expec_{\mathbb{P}} \insq{\mathrm{1}_{E}\inp{\prod_{i=1}^{\tau_{0,n}}\frac{\pZeroStar(X_i)}{\qZeroStar(X_i)}} \inp{\prod_{i=1}^{\tau_{0,n}}\frac{\qZeroStar(X_i)}{\pZeroStar(X_i)}}} \\	
	&\le 2^{- \alpha_1 n} \expec_{\mathbb{P}} \insq{\mathrm{1}_{E} 2^{S_{1,\tau_{1,n}}}} \\
	&\le 2^{- \alpha_1 n} \expec_{\mathbb{P}} \insq{2^{S_{1,\tau_{1,n}}}}.
	\end{align*}
	Using lemma \ref{lemma:ut_ub} with $\tau_1=\tau_{1,n}$, we have $\expec_{\mathbb{P}} \insq{2^{S_{1,\tau_{1,n}}}} \le 1$. 
	Thus, we have $\pi_{1|0}(\phi_n,\mathcal{A}_n) \le 2^{- \alpha_1 n}$ for all adversary strategies $\mathcal{A}_n$. Using lemma \ref{lemma:exp_tau_ub} with $\tau_0=\tau_{0,n}$, we have
    \begin{align*}
        \mathbb{E}_{\mathbb{P}}[\tau_n] \le \mathbb{E}_{\mathbb{P}}[\tau_{0,n}] &\le \frac{\expec_{\mathbb{P}} \insq{S_{0,\tau_{0,n}}}}{D(\pOneStar\|\qOneStar)} \\
        &\le \frac{\alpha_0 n + c_1}{D(\pOneStar\|\qOneStar)}.
    \end{align*}
    The final inequality follows from the fact that at the penultimate step the sum $S_{1,\tau_{1,n}-1}$ cannot exceed the threshold and the final increment $\log \frac{\pOneStar(X_{\tau_0})}{\qOneStar(X_{\tau_0})}$ is upper bounded by $c_0$.
    Thus,
	\begin{equation*}
	\liminf_{n \rightarrow \infty} -\frac{\log \sup \limits_{\mathcal{A}_n} \pi_{1|0}(\phi_n, \mathcal{A}_n)}{\sup \limits_{\mathcal{A}_n} \mathrm{E}_{\mathbb{P}}[\tau_n]} \ge \frac{\alpha_1}{\alpha_0}D(\pOneStar\|\qOneStar).
	\end{equation*}
	By a symmetric argument, we get
    \begin{equation*}
	\liminf_{n \rightarrow \infty} -\frac{\log \sup \limits_{\mathcal{A}_n} \pi_{0|1}(\phi_n, \mathcal{A}_n)}{\sup \limits_{\mathcal{A}_n} \mathrm{E}_{\mathbb{Q}}[\tau_n]} \ge \frac{\alpha_0}{\alpha_1}D(\qZeroStar\|\pZeroStar).
	\end{equation*}
    The above inequalities taken together complete the proof of achievability.
	\subsection{Proof of converse of Theorem~\ref{thrm:bemv1}}
	Assume that the pair $(E_0,E_1)$ is achievable such that $E_0>0,E_1>0$.	Consider the following adversary strategy, $\mathcal{A}_n'$. Fix $\hat{p}_i(.)=\pOneStar$ for all $i \in \mathbb{N}$, i.e. under $H_0$ the observations are drawn i.i.d. according to $\pOneStar$. Similarly, fix $\hat{q}_i(.)=\qOneStar$ for all $i \in \mathbb{N}$, i.e. under $H_1$ the observations are drawn according to $\qOneStar$. We follow a standard converse argument \cite{PolyanskiyWu24} which we repeat here for completeness. By data processing inequality (the processing is through the decision rule), we have
	\begin{equation} \label{eqn:dpi}
	D(\bern(\pi_{1|0}(\phi_n,\mathcal{A}_n'))\|\bern(1-\pi_{0|1}(\phi_n,\mathcal{A}_n'))) \le D(\mathbb{P}\|\mathbb{Q}).
	%\pi_{1|0}(\phi_n) \log \frac{\pi_{1|0}(\phi_n)}{1-\pi_{0|1}(\phi_n)} + 1-\pi_{0|1}(\phi_n) \log \frac{1-\pi_{1|0}(\phi_n)}{\pi_{0|1}(\phi_n)} &\le \expec_{\mathbb{P}_n}\insq{S_{1,\tauN}} \\
	\end{equation} 
	The L.H.S. can be lower bounded as follows
	\begin{align*}
	&D(\bern(\pi_{1|0}(\phi_n,\mathcal{A}_n'))\|\bern(1-\pi_{0|1}(\phi_n,\mathcal{A}_n'))) \\
	&= \pi_{1|0}(\phi_n,\mathcal{A}_n') \log \frac{\pi_{1|0}(\phi_n,\mathcal{A}_{n}')}{1 \!-\! \pi_{0|1}(\phi_n,\mathcal{A}_{n}')} \\
	&\hspace{12pt}+ \! (1 \!-\! \pi_{1|0}(\phi_n,\mathcal{A}_n')) \log \frac{1 \!-\! \pi_{1|0}(\phi_n,\mathcal{A}_n')}{\pi_{0|1}(\phi_n,\mathcal{A}_n')} \\
	&= -h(\pi_{1|0}(\phi_n,\mathcal{A}_n')) - \pi_{1|0}(\phi_n,\mathcal{A}_n') \log (1 - \pi_{0|1}(\phi_n,\mathcal{A}_n')) \\
	&\hspace{12pt} - (1-\pi_{1|0}(\phi_n,\mathcal{A}_n')) \log \pi_{0|1}(\phi_n,\mathcal{A}_n') \\
	&\ge -h(\pi_{1|0}(\phi_n,\mathcal{A}_n')) - (1-\pi_{1|0}(\phi_n,\mathcal{A}_n')) \log \pi_{0|1}(\phi_n,\mathcal{A}_n'),
	\end{align*}
	where the last inequality follows from the non-negativity of $-\pi_{1|0}(\phi_n,\mathcal{A}') \log (1-\pi_{0|1}(\phi_n,\mathcal{A}'))$. The R.H.S. can be written as follows
	\begin{align*}
	D(\mathbb{P}\|\mathbb{Q}) &= \expec_{\mathbb{P}}\insq{S_{0,\tauN}} \\
	&\overset{(a)}{=} \expec_{\mathbb{P}}[\tauN] D(\pOneStar\|\qOneStar)
	\end{align*} 
	The equality $(a)$ follows from the following: By Claims 1 and 2, $\{M_{0,t}\}_{t}$ and $\{\tilde{M}_{1,t}\}_{t}$ are martingales when the observations are i.i.d. according to $\pOneStar$. Thus, \eqref{eqn:submart} is an equality in this case. Since $\tau_n$ is a valid stopping rule, we have $\expec_{\mathbb{P}}[\tau_n] < \infty$. We conclude by applying the dominated convergence theorem.
	Applying these bounds to \eqref{eqn:dpi},
	\begin{align*}
	-\frac{1}{\expec_{\mathbb{P}}[\tauN]} \log \pi_{0|1}(\phi_n,\mathcal{A}_n') \le \frac{D(\pOneStar\|\qOneStar) + \frac{h(\pi_{1|0}(\phi_n,\mathcal{A}_n'))}{\expec_{\mathbb{P}}[\tauN]}}{1-\pi_{1|0}(\phi_n,\mathcal{A}_n')}.  
	\end{align*}
    Since $E_0 > 0$, $\pi_{1|0}(\phi_n,\mathcal{A}_n') \rightarrow 0$ as $n \rightarrow \infty$. Also, since $\expec_{\mathbb{P}}[\tauN]$ is unbounded as $n \rightarrow \infty$,  we have 
	\begin{equation} \label{eqn:conv_e1}
	\liminf_{n \rightarrow \infty} -\frac{1}{\expec_{\mathbb{P}}[\tauN]} \log \pi_{0|1}(\phi_n,\mathcal{A}_n') \le D(\pOneStar\|\qOneStar).
	\end{equation}
	Now, consider an alternate adversary strategy, $\mathcal{A}_n''$. Fix $\hat{p}_i(.)=\pZeroStar$ for all $i \in \mathbb{N}$, i.e. under $H_0$, the observations are i.i.d. according to $\pZeroStar$. Similarly, fix $\hat{q}_i(.)=\qZeroStar$ for all $i \in \mathbb{N}$, i.e. under $H_1$, the observations are i.i.d. according to $\qZeroStar$. By data processing inequality, we have
	\begin{equation}
	D(\bern(1-\pi_{0|1}(\phi_n,\mathcal{A}_n''))\|\bern(\pi_{1|0}(\phi_n,\mathcal{A}_n''))) \! \le \! D(\mathbb{Q}\|\mathbb{P}).
	\end{equation}
	Proceeding in a similar manner as before, if $E_1 > 0$,
	\begin{equation} \label{eqn:conv_e0}
	\liminf_{n \rightarrow \infty} -\frac{1}{\expec_{\mathbb{Q}}[\tauN]} \log \pi_{1|0}(\phi_n,\mathcal{A}_n'') \le D(\qZeroStar\|\pZeroStar).
	\end{equation}
	Taken together \eqref{eqn:conv_e1} and \eqref{eqn:conv_e0} imply that
    \begin{align*}
        %&\inp{\liminf_{n \rightarrow \infty} -\frac{\log \sup \limits_{\mathcal{A}_n} \pi_{1|0}(\phi_n, \mathcal{A}_n)}{\sup \limits_{\mathcal{A}_n} \mathrm{E}_{\mathbb{P}}[\tau_n]}} \inp{\liminf_{n \rightarrow \infty} -\frac{\log \sup \limits_{\mathcal{A}_n} \pi_{0|1}(\phi_n, \mathcal{A}_n)}{\sup \limits_{\mathcal{A}_n} \mathrm{E}_{\mathbb{Q}}[\tau_n]}} \\
        &\inp{\liminf_{n \rightarrow \infty} -\frac{\log \pi_{0|1}(\phi_n,\mathcal{A}_n')}{\expec_{\mathbb{P}}[\tauN]}} \inp{\liminf_{n \rightarrow \infty} -\frac{\log \pi_{1|0}(\phi_n,\mathcal{A}_n'')}{\expec_{\mathbb{Q}}[\tauN]}} \\
        &=\inp{\liminf_{n \rightarrow \infty} -\frac{\log \pi_{0|1}(\phi_n,\mathcal{A}_n')}{\expec_{\mathbb{Q}}[\tauN]}} \inp{\liminf_{n \rightarrow \infty} -\frac{\log \pi_{1|0}(\phi_n,\mathcal{A}_n'')}{\expec_{\mathbb{P}}[\tauN]}} \\
        &\le D(\pOneStar\|\qOneStar)D(\qZeroStar\|\pZeroStar).
    \end{align*}
    By the definition of achievability of a pair of exponents $(E_0,E_1)$, we have
    \begin{align*}
        E_0 &\le \liminf_{n \rightarrow \infty} -\frac{\log \sup_{\mathcal{A}_n} \pi_{1|0}(\phi_n, \mathcal{A}_n)}{\sup_{\mathcal{A}_n} \mathrm{E}_{\mathbb{P}}[\tau_n]} \\
        &\le \liminf_{n \rightarrow \infty} -\frac{\log \pi_{1|0}(\phi_n,\mathcal{A}_n'')}{\expec_{\mathbb{P}}[\tauN]},
        %E_1 &\le \liminf_{n \rightarrow \infty} -\frac{\log \sup_{\mathcal{A}_n} \pi_{0|1}(\phi_n, \mathcal{A}_n)}{\sup_{\mathcal{A}_n} \mathrm{E}_{\mathbb{Q}}[\tau_n]},
    \end{align*}
    and
    \begin{align*}
        E_1 &\le \liminf_{n \rightarrow \infty} -\frac{\log \sup_{\mathcal{A}_n} \pi_{0|1}(\phi_n, \mathcal{A}_n)}{\sup_{\mathcal{A}_n} \mathrm{E}_{\mathbb{Q}}[\tau_n]}, \\
        &\le \liminf_{n \rightarrow \infty} -\frac{\log \pi_{0|1}(\phi_n,\mathcal{A}_n')}{\expec_{\mathbb{Q}}[\tauN]}.
    \end{align*}
    Thus, $E_0 E_1 \le D(\pOneStar\|\qOneStar)D(\qZeroStar\|\pZeroStar)$.
%\input{variations}
%!TeX root=seq_adv_hypo_test.tex
\section{Probability constraint on the stopping time} \label{app:prob_constraint}
We now study the problem under the constraint that the stopping time will not exceed $n$ with high probability as in \cite{pan2022asymptotics}. A pair of exponents $(\bar{E}_0,\bar{E}_1)$ is said to be \emph{achievable}, if there exists a sequence of tests $(\phi_n=(\tauN,Z_n))_{n \in \mathbb{N}}$ such that
\begin{align*}
\bar{E}_0 &\le \liminf_{n \rightarrow \infty} -\frac{1}{n} \log \sup_{\mathcal{A}_n} \pi_{1|0}(\phi_n, \mathcal{A}_n) \\
\bar{E}_1 &\le \liminf_{n \rightarrow \infty} -\frac{1}{n} \log \sup_{\mathcal{A}_n} \pi_{0|1}(\phi_n, \mathcal{A}_n),
\end{align*}
subject to the following constraint: for any $0 < \epsilon < 1$, there exists an integer $n_0(\epsilon)$ such that for all $n > n_0(\epsilon)$, the stopping time $\tauN$ satisfies the following constraint,
\begin{equation}
\sup_{\mathcal{A}_n} \mathbb{P}(\tauN > n) < \epsilon \hspace{10pt} \textup{and} \hspace{10pt} \sup_{\mathcal{A}_n} \mathbb{Q}(\tauN > n) < \epsilon.
\end{equation}
Define $\bar{\mathcal{E}}(\mathcal{P},\mathcal{Q})$ to be the set of all achievable $(\bar{E}_0,\bar{E}_1)$ pairs. Notice that this is a closed set. We want to characterize this region.
\begin{theorem} \label{thrm:bemv2}
	For sets $\sP,\sQ \subseteq \mathbb{R}^{\sX}$ which satisfy the properties A1 and A2, the closure of the set of achievable pairs of exponents is
	\begin{equation}
	\bar{\mathcal{E}}(\mathcal{P},\mathcal{Q}) = \inb{ (\bar{E}_0,\bar{E}_1):\bar{E}_0 \le D(\qZeroStar\|\pZeroStar), \bar{E}_1 \le D(\pOneStar\|\qOneStar)}.
	\end{equation}
	Furthermore, the corner point $(D(\pOneStar\|\qOneStar),D(\qZeroStar\|\pZeroStar))$ is achievable.
\end{theorem}

\subsection{Proof of achievability in Theorem~\ref{thrm:bemv2}}
	The stopping time is given by,
	\begin{equation} \label{eqn:tauN_prob_constraint}
	\tau_n := \inf \{t \in \mathbb{N}: S_{1,t} \geq n(D(\pOneStar\|\qOneStar)-\delta) \hspace{5pt} \textup{or} \\ \hspace{5pt} S_{0,t} \geq n(D(\qZeroStar\|\pZeroStar)-\delta)\},
	\end{equation}
	where $\delta > 0$. Let $T_1:=n(D(\pOneStar\|\qOneStar)-\delta_1)$, $T_2 := n(D(\qZeroStar\|\pZeroStar)-\delta_0)$ be the two thresholds. The decision rule is defined to be
	\begin{equation} \label{eqn:decision_rule_prob_constraint}
	Z_n(x^{\tauN}) := \left\{
	\begin{array}{ll}
	0 & \mbox{if } S_{1,\tauN} \geq T_1, S_{0,\tauN} < T_2\\
	1 & \mbox{if } S_{0,\tauN} \geq T_2.
	\end{array}
	\right.
	\end{equation}.
	Recall that $\mathcal{A}_n$ denote the adversary strategy $(\hat{p}_t,\hat{q})_{t \in \mathbb{N}}$. We need to show that for every $\epsilon \in (0,1)$, there exists an integer $n_0(\epsilon)$ such that,\\
	(i) $\sup \limits_{\mathcal{A}_n} \mathbb{P}(\tauN > n) < \epsilon$ and $\sup \limits_{\mathcal{A}_n} \mathbb{Q}(\tauN > n) < \epsilon$.\\
	(ii) \begin{align*}
	\liminf_{n \rightarrow \infty} -\frac{1}{n} \log \sup_{\mathcal{A}_n} \pi_{1|0}(\phi_n, \mathcal{A}_n) &= D(\qZeroStar\|\pZeroStar) \\
	\liminf_{n \rightarrow \infty} -\frac{1}{n} \log \sup_{\mathcal{A}_n} \pi_{0|1}(\phi_n, \mathcal{A}_n) &= D(\pOneStar\|\qOneStar).
	\end{align*}
	\textit{Proof of (i)}: We will show $\sup \limits_{\mathcal{A}_n} \mathbb{P}(\tauN > n) < \epsilon$, the other case follows by symmetry. Define $\tau_{0,n}$ and $\tau_{1,n}$ to be
	\begin{align}
	\tau_{0,n} &:= \inf \inb{t \in \mathbb{N}: S_{0,t} \geq n(D(\qZeroStar\|\pZeroStar)-\delta)} \label{eqn:tau0N_prob_constraint} \\
	\tau_{1,n} &:= \inf \inb{t \in \mathbb{N}: S_{1,t} \geq n(D(\pOneStar\|\qOneStar)-\delta)}. \label{eqn:tau1N_prob_constraint}
	\end{align}
	Observe that $\tauN=\tau_{0,n} \land \tau_{1,n}$. Thus, $\mathbb{P}(\tauN > n) \le \mathbb{P}(\tauZeroN > n)$. It suffices to bound $\mathbb{P}(\tauZeroN > n)$.
	\begin{align*}
	\mathbb{P}(\tauZeroN > n) &= \mathbb{P} \insq{\cap_{j=1}^{n} \{S_{0,j} < n(D(\qZeroStar\|\pZeroStar)-\delta) \}} \\
	&\le \mathbb{P} \insq{S_{0,n} < n(D(\qZeroStar\|\pZeroStar)-\delta)} \\
	&= \mathbb{P} \insq{M_{0,n} < -n\delta} \\
	&\overset{(a)}{\le} \exp \insq{-\frac{n^2\delta^2}{2\sum_{k=1}^{n}(2c_0)^2}} \\
	&= \exp \insq{-\frac{n\delta^2}{8c_0^2}}
	\end{align*}
	The inequality (a) follows from the fact that $M_{0,t}$ is a submartingale (Claim 1) and Azuma-Hoeffding inequality for bounded differences. Thus, for every $\epsilon \in (0,1)$, there exists an integer $n(\epsilon,\mathcal{A}_n)$ such that the R.H.S. in the above equation is less than $\epsilon$. We can pick $n_0(\epsilon)$ to be $\sup \limits_{\mathcal{A}_n} n(\epsilon,\mathcal{A}_n)$. 
	\\
	\textit{Proof of (ii)}: Let $E \subseteq \mathcal{F}_{\tau_n}$ be the error event under $H_0$. Thus, $E = \{ S_{0,\tau_n} \ge n(D(\qZeroStar\|\pZeroStar)-\delta)\} = \{ \tau_{n} = \tau_{0,n}\}$.
	\begin{align*}
	\pi_{1|0}(\phi_n,\mathcal{A}_n) &= \mathbb{P} (E) \\
	&= \expec_{\mathbb{P}}[\mathrm{1}_{E}] \\
	&= \expec_{\mathbb{P}}\insq{\mathrm{1}_{E} 2^{-S_{0,\tau_{0,n}}} 2^{S_{0,\tau_{0,n}}}} \\
	%&= \expec_{\mathbb{P}} \insq{\mathrm{1}_{E}\inp{\prod_{i=1}^{\tau_{0,n}}\frac{\pZeroStar(X_i)}{\qZeroStar(X_i)}} \inp{\prod_{i=1}^{\tau_{0,n}}\frac{\qZeroStar(X_i)}{\pZeroStar(X_i)}}} \\	
	&\le 2^{-n(D(\qZeroStar\|\pZeroStar)-\delta)} \expec_{\mathbb{P}} \insq{\mathrm{1}_{E} 2^{S_{0,\tau_{0,n}}}} \\
	&\le 2^{-n(D(\qZeroStar\|\pZeroStar)-\delta)} \expec_{\mathbb{P}} \insq{2^{S_{0,\tau_{0,n}}}}.
	\end{align*}
	Using lemma \ref{lemma:ut_ub} with $\tau_0=\tau_{0,n}$, we have $\expec_{\mathbb{P}} \insq{2^{S_{0,\tau_{0,n}}}} \le 1$. Thus, we have $\pi_{1|0}(\phi_n,\mathcal{A}_n) \le 2^{-n(D(\qZeroStar\|\pZeroStar)-\delta)}$ for all adversary strategies $\mathcal{A}_n$. This implies
	\begin{equation*}
	\liminf_{n \rightarrow \infty} -\frac{1}{n} \log \sup_{\mathcal{A}_n} \pi_{1|0}(\phi_n, \mathcal{A}_n) \ge D(\qZeroStar\|\pZeroStar).
	\end{equation*}
	The analysis for type-II error probability follows from symmetry.
	\subsection{Proof of converse in Theorem~\ref{thrm:bemv2}}
	We use a lemma from \cite{li2020second}.
	\begin{lemma}[{\cite[Lemma 3]{li2020second}}] \label{lemma:LiTan}
		Consider the binary hypothesis testing probem- $H_0:p$ vs $H_1:q$. Let $\phi=(\tau,Z)$ be a sequential test. Let $\mathbb{P}$ and $\mathbb{Q}$ be the i.i.d. probability measures on $\mathcal{F}_{\tau}$ associated with $p$ and $q$ respectively. Suppose $\mathbb{P}[\tau < \infty] = 1$ and $\mathbb{Q}[\tau < \infty] = 1$. Then, for any event $F \in \mathcal{F}_{\tau}$, $\lambda > 0$ we have
		\begin{align}
			\mathbb{P}[F] - \lambda \mathbb{Q}[F] &\le \mathbb{P}\insq{\sum_{i=1}^{\tau} \frac{q(X_i)}{p(X_i)} \le -\log \lambda} \label{eqn:LiTan_1} \\
			\mathbb{Q}[F] - \frac{1}{\lambda} \mathbb{P}[F] &\le \mathbb{Q}\insq{\sum_{i=1}^{\tau} \frac{q(X_i)}{p(X_i)} \ge - \log \lambda}. \label{eqn:LiTan_2}
		\end{align}
	\end{lemma}
	Assume that the pair $(\bar{E}_0,\bar{E}_1)$ is achievable by the family of tests $(\phi_n=(\tauN,Z_n))_n$. $(\phi_n)_n$ satisfy the constraint that for any $\epsilon \in (0,1)$, there exists a large enough $n$ such that $\sup \limits_{\mathcal{A}_n} \mathbb{P}(\tauN > n) < \epsilon$ and $\sup \limits_{\mathcal{A}_n} \mathbb{Q}(\tauN > n) < \epsilon$. Consider the following adversary strategy, $\mathcal{A}_n$. Fix $\hat{p}_i(.)=\pOneStar$ for all $i \in \mathbb{N}$, i.e. under $H_0$ the observations are drawn i.i.d. according to $\pOneStar$. Similarly, fix $\hat{q}_i(.)=\qOneStar$ for all $i \in \mathbb{N}$, i.e. under $H_1$ the observations are drawn according to $\qOneStar$. We invoke equation \ref{eqn:LiTan_2}, lemma \ref{lemma:LiTan} with $p = \pZeroStar$, $q = \qZeroStar$. Let $E \subseteq \mathcal{F}_{\tau_n}$ be the error event under $H_0$. Thus, $E = \{ Z_n = 1\}$. Recall that $S_{0,\tauN} = \sum_{i=1}^{\tauN} \frac{\qZeroStar(X_i)}{\pZeroStar(X_i)}$.
	\begin{align*}
	\mathbb{Q}[E] - \frac{1}{\lambda} \mathbb{P}[E] &\le \mathbb{Q}\insq{S_{0,\tauN} \ge -\log \lambda}
	\end{align*}
	Thus, we have
	\begin{align*}
	1-\mathbb{Q}[E^c] - \frac{1}{\lambda} \mathbb{P}[E] &\le \mathbb{Q}\insq{S_{0,\tauN} \ge - \log \lambda, \tauN \le n} \\
	&\hspace{10pt}+ \mathbb{Q}\insq{S_{0,\tauN} \ge - \log \lambda,\tauN > n} \\
	&\le \mathbb{Q}\insq{S_{0,\tauN} \ge - \log \lambda, \tauN \le n}\\
	& \hspace{10pt} + \mathbb{Q}\insq{\tauN > n}.
	\end{align*}
	Rearranging the terms, we get
	\begin{align} 
	\log \mathbb{P}[E] \ge &\log \bigg[ \lambda \bigg(1 - \mathbb{Q}[E^c] - \mathbb{Q}\insq{\tauN > n} \nonumber \\ 
	&- \mathbb{Q}\insq{S_{0,\tauN} \ge - \log \lambda, \tauN \le n} \bigg) \bigg]. \label{eqn:prob_constraint_conv_ineq}
	\end{align}
	Let $\delta > 0$ and $\log \lambda = -n(D(\qZeroStar\|\pZeroStar) + \delta)$. We bound the following term
	\begin{align}
	\mathbb{Q}\insq{S_{0,\tauN} \! \ge \! - \! \log \lambda, \! \tauN \! \le \! n} &\le \mathbb{Q} \! \insq{\bigcap \limits_{i=1}^{n} \{S_{0,i} \! \ge \! n(D(\qZeroStar\|\pZeroStar) \! + \! \delta)\}} \nonumber\\
	& \le \mathbb{Q}\insq{\{S_{0,n} \! \ge \! n(D(\qZeroStar\|\pZeroStar) \! + \! \delta)\}} \nonumber\\
	&= \mathbb{Q}\insq{\{M_{0,n} \! \ge \! n\delta\}} \nonumber\\
	&\overset{(a)}{\le} \exp \insq{-\frac{n^2\delta^2}{2\sum_{k=1}^{n}(2c_0)^2}} \nonumber\\
	&= \exp \insq{-\frac{n\delta^2}{8c_0^2}} \label{eqn:prob_constraint_conv_term1}
	\end{align}
	$(a)$ follows from Claim 1 and Azuma-Hoeffding inequality for submartingale sequences with bounded differences. Thus,
	\begin{equation} \label{eqn:prob_constraint_conv_term1_limit}
	\lim_{n \rightarrow \infty} \mathbb{Q}\insq{S_{0,\tauN} \! \ge \! - \! \log \lambda, \! \tauN \! \le \! n} = 0
	\end{equation}
	From \eqref{eqn:prob_constraint_conv_ineq}, \eqref{eqn:prob_constraint_conv_term1}, we have for large enough $n$, 
	\begin{align*}
	-\frac{1}{n} \log \mathbb{P}[E] \le &D(\qZeroStar\|\pZeroStar) + \delta \\
	&- \frac{1}{n} \log \bigg(1 - \mathbb{Q}[E^c] - \epsilon - \exp \insq{-\frac{n\delta^2}{8c_0^2}} \bigg)
	\end{align*}
	Since $\bar{E}_1 > 0$, $\lim \limits_{n \rightarrow \infty} \mathbb{Q}[E^c] = 0$. Using \eqref{eqn:prob_constraint_conv_term1_limit}, we get
	\begin{equation} \label{eqn:prob_constraint_conv_e0}
	\bar{E}_0 = \lim_{n \rightarrow \infty} \frac{1}{n} \log \mathbb{P}[E] = D(\qZeroStar\|\pZeroStar)+\delta.
	\end{equation}
	Now, consider an alternate adversary strategy, $\mathcal{A}_n'$. Fix $\hat{p}_i(.)=\pOneStar$ for all $i \in \mathbb{N}$, i.e. under $H_0$ the observations are drawn i.i.d. according to $\pOneStar$. Similarly, fix $\hat{q}_i(.)=\qOneStar$ for all $i \in \mathbb{N}$, i.e. under $H_1$ the observations are drawn according to $\qOneStar$. We invoke equation \ref{eqn:LiTan_1}, lemma \ref{lemma:LiTan} with $p = \pOneStar$, $q = \qOneStar$. Let $E' \subseteq \mathcal{F}_{\tau_n}$ be the error event under $H_0$. Thus, $E' = \{ Z_n = 0\}$. Repeating the same steps, we get
	\begin{equation} \label{eqn:prob_constraint_conv_e1}
	\bar{E}_1 = \lim_{n \rightarrow \infty} \frac{1}{n} \log \mathbb{Q}[E'] = D(\pOneStar\|\qOneStar)+\delta.
	\end{equation}
	We now let $\delta \rightarrow 0$. Taken together \eqref{eqn:prob_constraint_conv_e0} and \eqref{eqn:prob_constraint_conv_e1} imply that for all pairs $(\bar{E}_0,\bar{E}_1)$ such that $\bar{E}_0,\bar{E}_1 > 0$, $\bar{E}_0 \le D(\qZeroStar\|\pZeroStar), \bar{E}_1 \le D(\pOneStar\|\qOneStar)$.

\begin{comment}
    
    \begin{figure}
        \centering
        %\includegraphics[width=0.5\linewidth]{}
        \tikzset{every picture/.style={line width=0.75pt}} %set default line width to 0.75pt        

\begin{tikzpicture}[x=0.75pt,y=0.75pt,yscale=-1,xscale=1]
%uncomment if require: \path (0,384); %set diagram left start at 0, and has height of 384

%Shape: Axis 2D [id:dp8815546877731582] 
\draw  (140.5,323.4) -- (509.5,323.4)(177.4,48) -- (177.4,354) (502.5,318.4) -- (509.5,323.4) -- (502.5,328.4) (172.4,55) -- (177.4,48) -- (182.4,55)  ;
%Straight Lines [id:da7707732658243266] 
\draw    (177.5,176) -- (363.5,174) ;
%Straight Lines [id:da6042086790247282] 
\draw    (364.5,174) -- (364.5,325) ;

% Text Node
\draw (465,348.4) node [anchor=north west][inner sep=0.75pt]  [font=\large]  {$E_{1}$};
% Text Node
\draw (142,75.4) node [anchor=north west][inner sep=0.75pt]  [font=\large]  {$E_{0}$};
% Text Node
\draw (83,163.4) node [anchor=north west][inner sep=0.75pt]  [font=\large]  {$D\left( q_{1}^{*} \| p_{1}^{*}\right)$};
% Text Node
\draw (324,335.4) node [anchor=north west][inner sep=0.75pt]  [font=\large]  {$D\left( p_{0}^{*} \| q_{0}^{*}\right)$};
% Text Node
\draw (373,140.4) node [anchor=north west][inner sep=0.75pt]  [font=\large]  {$\left( D\left( p_{0}^{*} \| q_{0}^{*}\right) ,D\left( q_{1}^{*} \| p_{1}^{*}\right)\right)$};

\end{tikzpicture}
        \caption{Caption}
        \label{fig:prob_constraint}
    \end{figure}

\end{comment}
%!TeX root=seq_adv_hypo_test.tex

\section{Constraint on the probability of error} \label{app:dual_problem}
Alternatively, we can put a constraint on the type-I and type-II errors.
A pair of exponents $(\tilde{E}_0,\tilde{E}_1)$ is said to be \emph{achievable}, if there exists a sequence of tests $(\phi_{\beta}=(\tau_{\beta},Z_{\beta}))_{\beta}$ such that
\begin{align*}
\tilde{E}_0 &\le \liminf_{\beta \rightarrow 0} \inf_{\mathcal{A}_\beta} \frac{-\log \beta}{\expec_{\mathbb{P}}[\tau_{\beta}]} \\
\tilde{E}_1 &\le \liminf_{\beta \rightarrow 0} \inf_{\mathcal{A}_\beta} \frac{-\log \beta}{\expec_{\mathbb{Q}}[\tau_{\beta}]}
\end{align*}
and $\sup_{\mathcal{A}_{\beta}} \pi_{1|0}(\phi_{\beta},\mathcal{A}_{\beta}) \le \beta$, $\sup_{\mathcal{A}_{\beta}} \pi_{0|1}(\phi_{\beta},\mathcal{A}_{\beta}) \le \beta$. In this setting, we have the following result.
\begin{theorem} \label{thrm:bemv3}
	For sets $\sP,\sQ \subseteq \mathbb{R}^{\sX}$ which satisfy the properties A1 and A2, the closure of the set of achievable pairs of exponents is
	\begin{equation}
	\tilde{\mathcal{E}}(\mathcal{P},\mathcal{Q}) = \inb{ (\tilde{E}_0,\tilde{E}_1):\tilde{E}_0 \le D(\pOneStar\|\qOneStar), \tilde{E}_1 \le D(\qZeroStar\|\pZeroStar)}.
	\end{equation}
	Furthermore, the corner point $(D(\pOneStar\|\qOneStar),D(\qZeroStar\|\pZeroStar))$ is achievable.
\end{theorem}

\subsection{Proof of achievability Theorem~\ref{thrm:bemv3}}
The stopping time is given by,
\begin{equation} \label{eqn:tau_dual_problem}
\tau_\beta := \inf \{t \in \mathbb{N}: S_{1,t} \geq - \log \beta \hspace{5pt} \textup{or} \hspace{5pt} S_{0,t} \geq -\log \beta\},
\end{equation}
where $\beta > 0$. Let $T_1:=-\log \beta$, $T_2 := -\log \beta$ be the two thresholds. The decision rule is defined to be
\begin{equation} \label{eqn:decision_rule_dual_problem}
Z_{\beta}(x^{\tau_\beta}) := \left\{
\begin{array}{ll}
0 & \mbox{if } S_{1,\tau_\beta} \geq T_1, S_{0,\tau_\beta} < T_2\\
1 & \mbox{if } S_{0,\tau_\beta} \geq T_2.
\end{array}
\right.
\end{equation}
We define $\tau_{0,\beta}$ and $\tau_{1,\beta}$ to be
\begin{align}
\tau_{0,\beta} &:= \inf \inb{t \in \mathbb{N}: S_{0,t} \geq -\log \beta} \label{eqn:tau0Beta} \\
\tau_{1,\beta} &:= \inf \inb{t \in \mathbb{N}: S_{1,t} \geq -\log \beta}. \label{eqn:tau1Beta}
\end{align}
Recall that $\mathcal{A}_\beta$ denote the adversary strategy $(\hat{p}_t,\hat{q})_{t \in \mathbb{N}}$. We need to show that \\
(i) $\sup \limits_{\mathcal{A}_\beta} \mathbb{P}(Z_\beta = 1) < \beta$ and $\sup \limits_{\mathcal{A}_\beta} \mathbb{Q}(Z_\beta = 0) < \beta$.\\
(ii) \begin{align*}
\liminf_{\beta \rightarrow 0} \inf_{\mathcal{A}_\beta} \frac{-\log \beta}{\expec_{\mathbb{P}}[\tau_{\beta}]} &= D(\pOneStar\|\qOneStar)\\
\liminf_{\beta \rightarrow 0} \inf_{\mathcal{A}_\beta} \frac{-\log \beta}{\expec_{\mathbb{Q}}[\tau_{\beta}]} &= D(\qZeroStar\|\pZeroStar).
\end{align*}
\textit{Proof of (i)}: Let $E \subseteq \mathcal{F}_{\tau_\beta}$ be the error event under $H_0$. Thus, $E = \{ S_{0,\tau_\beta} \ge -\log \beta \} = \{ \tau_{\beta} = \tau_{0,\beta}\}$.
\begin{align*}
\pi_{1|0}(\phi_\beta,\mathcal{A}_\beta) &= \mathbb{P} (E) \\
&= \expec_{\mathbb{P}}[\mathrm{1}_{E}] \\
&= \expec_{\mathbb{P}}\insq{\mathrm{1}_{E} 2^{-S_{0,\tau_{0,\beta}}} 2^{S_{0,\tau_{0,\beta}}}} \\
%&= \expec_{\mathbb{P}} \insq{\mathrm{1}_{E}\inp{\prod_{i=1}^{\tau_{0,n}}\frac{\pZeroStar(X_i)}{\qZeroStar(X_i)}} \inp{\prod_{i=1}^{\tau_{0,n}}\frac{\qZeroStar(X_i)}{\pZeroStar(X_i)}}} \\	
&\le 2^{\log \beta} \expec_{\mathbb{P}} \insq{\mathrm{1}_{E} 2^{S_{0,\tau_{0,\beta}}}} \\
&\le \beta \cdot \expec_{\mathbb{P}} \insq{2^{S_{0,\tau_{0,\beta}}}}.
\end{align*}
Using lemma \ref{lemma:ut_ub} with $\tau_0=\tau_{0,\beta}$, we have $\expec_{\mathbb{P}} \insq{2^{S_{0,\tau_{0,\beta}}}} \le 1$. Thus, $\pi_{1|0}(\phi_n,\mathcal{A}_n) \le \beta$ for all adversary strategies $\mathcal{A}_n$. Likewise for $\pi_{0|1}(\phi_n,\mathcal{A}_n)$. \\
\textit{Proof of (ii)}: Using lemma \ref{lemma:exp_tau_ub} with $\tau_1=\tau_{1,\beta}$, we have
\begin{equation*}
\mathbb{E}_{\mathbb{P}}[\tau_{1,\beta}] \le \frac{\expec_{\mathbb{P}} \insq{S_{1,\tau_{1,\beta}}}}{D(\pOneStar\|\qOneStar)}.
\end{equation*}
Now,
\begin{align*}
S_{1,\tau_{1,\beta}} &= S_{1,\tau_{1,\beta}-1} + \log \frac{\pOneStar(X_{\tau_{1,\beta}})}{\qOneStar(X_{\tau_{1,\beta}})} \\
&\le -\log \beta + c_1.
\end{align*}
The inequality above follows from the fact that at the penultimate step $S_{1,\tau_{1,\beta}-1}$ cannot be above its threshold and the final increment can be at most $\max \limits_{x \in \sX}|\log \frac{\pOneStar(x)}{\qOneStar(x)}| = c_1$. Thus, we have
\begin{align*}
\mathbb{E}_{\mathbb{P}}[\tau_{1,\beta}] \le \frac{-\log \beta + c_1}{D(\pOneStar\|\qOneStar)}.
\end{align*}
Taking the limit $\beta \rightarrow \infty$, we get
\begin{align*}
    \lim_{\beta \rightarrow 0} \frac{-\log \beta}{\expec_{\mathbb{P}}[\tau_\beta]} \ge D(\pOneStar\|\qOneStar)
\end{align*}
Likewise, we get
\begin{align*}
    \lim_{\beta \rightarrow 0} \frac{-\log \beta}{\expec_{\mathbb{Q}}[\tau_\beta]} \ge D(\qZeroStar\|\pZeroStar)
\end{align*}.
This completes the achievability.

\subsection{Proof of converse of Theorem~\ref{thrm:bemv3}}
Assume that the pair $(\tilde{E}_0,\tilde{E}_1)$ is achievable by the family of tests $(\phi_\beta=(\tau_\beta,Z_\beta))_\beta$. The test $\phi_\beta$ satisfies the constraint that $\sup \limits_{\mathcal{A}_\beta} \mathbb{P}(Z_\beta = 1) < \beta$ and $\sup \limits_{\mathcal{A}_\beta} \mathbb{Q}(Z_\beta = 0) < \beta$ where the supremum is over all possible adversary strategies. Consider the following adversary strategy, $\mathcal{A}_\beta'$. Fix $\hat{p}_i(.)=\pOneStar$ for all $i \in \mathbb{N}$, i.e. under $H_0$ the observations are drawn i.i.d. according to $\pOneStar$. Similarly, fix $\hat{q}_i(.)=\qOneStar$ for all $i \in \mathbb{N}$, i.e. under $H_1$ the observations are drawn according to $\qOneStar$. By data processing inequality we have,
\begin{align*}
	D(\bern(\pi_{1|0}(\phi_n,\mathcal{A}_{\beta}'))\|\bern(1-\pi_{0|1}(\phi_n,\mathcal{A}_{\beta}'))) &\le D(\mathbb{P}\|\mathbb{Q}) \\
    &= \expec_{\mathbb{P}}[\tau_\beta] D(\pOneStar\|\qOneStar).
\end{align*}
The L.H.S. can be lower bounded as follows.    
\begin{align*}
&D(\bern(\pi_{1|0}(\phi_\beta,\mathcal{A}_{\beta}'))\|\bern(1-\pi_{0|1}(\phi_\beta,\mathcal{A}_{\beta}'))) \\
&= \pi_{1|0}(\phi_\beta,\mathcal{A}_{\beta}') \log \frac{\pi_{1|0}(\phi_\beta,\mathcal{A}_{\beta}')}{1 \!-\! \pi_{0|1}(\phi_\beta,\mathcal{A}_{\beta}')} + \! (1 \!-\! \pi_{1|0}(\phi_\beta,\mathcal{A}_{\beta}')) \log \frac{1 \!-\! \pi_{1|0}(\phi_\beta,\mathcal{A}_{\beta}')}{\pi_{0|1}(\phi_\beta,\mathcal{A}_{\beta}')} \\
&= -h(\pi_{1|0}(\phi_n,\mathcal{A}_n')) - \pi_{1|0}(\phi_n,\mathcal{A}_n') \log (1 - \pi_{0|1}(\phi_n,\mathcal{A}_n')) - (1 \!-\! \pi_{1|0}(\phi_\beta,\mathcal{A}_{\beta}')) \log \pi_{0|1}(\phi_\beta,\mathcal{A}_{\beta}') \\
&\ge -h(\pi_{1|0}(\phi_n,\mathcal{A}_n')) - (1 \!-\! \beta) \log \beta.
\end{align*}
The last inequality is obtained by dropping a non-negative term and using the fact that \\
$\pi_{0|1}(\phi_\beta,\mathcal{A}_{\beta}') \le \beta,\pi_{0|1}(\phi_\beta,\mathcal{A}_{\beta}') \le \beta$. Thus, we have
\begin{align*}
    &\frac{-\log \beta}{\expec_{\mathbb{P}}[\tau_\beta]} \le  \frac{D(\pOneStar\|\qOneStar)}{1-\beta} + \frac{h(\pi_{1|0}(\phi_n,\mathcal{A}_n'))}{(1-\beta)\expec_{\mathbb{P}}[\tau_\beta]} \\
    &\lim_{\beta \rightarrow 0} \frac{-\log \beta}{\expec_{\mathbb{P}}[\tau_\beta]} \le D(\pOneStar\|\qOneStar) 
\end{align*}
Now, consider the adversary strategy, $\mathcal{A}_\beta''$. Fix $\hat{p}_i(.)=\pZeroStar$ for all $i \in \mathbb{N}$. Similarly, fix $\hat{q}_i(.)=\qOneStar$ for all $i \in \mathbb{N}$. By following similar arguments as before, we get
\begin{align*}
    \lim_{\beta \rightarrow 0} \frac{-\log \beta}{\expec_{\mathbb{Q}}[\tau_\beta]} \le D(\qZeroStar\|\pZeroStar) 
\end{align*}
This completes the proof.

\begin{comment}
    \tikzset{every picture/.style={line width=0.75pt}} %set default line width to 0.75pt        

\begin{tikzpicture}[x=0.75pt,y=0.75pt,yscale=-1,xscale=1]
%uncomment if require: \path (0,384); %set diagram left start at 0, and has height of 384

%Shape: Axis 2D [id:dp8815546877731582] 
\draw  (140.5,323.4) -- (509.5,323.4)(177.4,48) -- (177.4,354) (502.5,318.4) -- (509.5,323.4) -- (502.5,328.4) (172.4,55) -- (177.4,48) -- (182.4,55)  ;
%Straight Lines [id:da7707732658243266] 
\draw    (177.5,176) -- (363.5,174) ;
%Straight Lines [id:da6042086790247282] 
\draw    (364.5,174) -- (364.5,325) ;

% Text Node
\draw (465,348.4) node [anchor=north west][inner sep=0.75pt]  [font=\large]  {$E_{1}$};
% Text Node
\draw (142,75.4) node [anchor=north west][inner sep=0.75pt]  [font=\large]  {$E_{0}$};
% Text Node
\draw (83,163.4) node [anchor=north west][inner sep=0.75pt]  [font=\large]  {$D\left( p_{0}^{*} \| q_{0}^{*}\right)$};
% Text Node
\draw (324,335.4) node [anchor=north west][inner sep=0.75pt]  [font=\large]  {$D\left( q_{1}^{*} \| p_{1}^{*}\right)$};
% Text Node
\draw (373,140.4) node [anchor=north west][inner sep=0.75pt]  [font=\large]  {$\left( D\left( q_{1}^{*} \| p_{1}^{*}\right) ,D\left( p_{0}^{*} \| q_{0}^{*}\right)\right)$};

\end{tikzpicture}

\end{comment}
\section{Discussion}
In this work, we have focused on the first-order asymptotics of the error exponents under various settings. Characterizing second-order asymptotics can be a topic for future studies.
\appendices
\section{Tradeoff in the fixed length setting} \label{app:tradeoff}
A fixed length test can be written as $\phi_n=(n,Z_n)$, i.e. a sequential test with $\tau_n=n$. The adversary strategy $\mathcal{A}_n$ for the test $\phi_n$ is given by $(\hat{p}_t,\hat{q}_t)_{t=1}^{n}$ where $\hat{p}_t:\mathcal{X}^{t-1} \rightarrow \mathcal{P}$ and $\hat{q}_t:\mathcal{X}^{t-1} \rightarrow \mathcal{Q}$. Recall that the law of $X^t$ under $H_0$ is given by $\mathbb{P}(x^t)=\prod_{i=1}^{t}p(x_i|x^{i-1})$ where $p(.|x^{i-1})=\hat{p}_i(x^{i-1})$ and under $H_1$ by $\mathbb{Q}(x^t)=\prod_{i=1}^{t}q(x_i|x^{i-1})$ where $q(.|x^{i-1})=\hat{q}_i(x^{i-1})$. For a particular adversary strategy, the type-I and type-II errors are given by
\begin{align*}
    \pi_{1|0}(\phi_n,\mathcal{A}_n) &:= \mathbb{P}(Z_n = 1) \\
    \pi_{0|1}(\phi_n,\mathcal{A}_n) &:= \mathbb{Q}(Z_n = 0)
\end{align*}
We seek to characterize the maximum achievable type-I error exponent when the type-II error exponent is at least $r$. Formally, we define the Hoeffding exponent for the adversarial hypothesis testing problem to be
\begin{equation*}
    \expHoeffding(\mathcal{P},\mathcal{Q}) := \sup_{\phi_n} \left\{ \lim_{n \rightarrow \infty} \frac{-1}{n}  \sup_{\mathcal{A}_n} \log \pi_{1|0}(\phi_n,\mathcal{A}_n) \ \biggr | \lim_{n \rightarrow \infty} \frac{-1}{n} \sup_{\mathcal{A}_n} \log \pi_{0|1}(\phi_n,\mathcal{A}_n) \ge  r \right \}  
\end{equation*}
For distributions $p$ and $q$, let 
\begin{equation*}
    \psi_{\lambda}(p\|q) :=  \log \sum_{x \in \Omega} p(x)^{1-\lambda}q(x)^{\lambda}.
\end{equation*}
Observe that $\psi_{\lambda}(p\|q) = - \lambda D_{1-\lambda}(p\|q)$, where $D_{1-\lambda}(p\|q)$ is the R\'enyi divergence of order $1-\lambda$ between $p$ and $q$. In the case of simple binary hypothesis testing (i.e. when $\mathcal{P}=\{p\}$ and $\mathcal{Q}=\{q\}$), Blahut \cite{blahut1974hypothesis} characterized the Hoeffding exponent as follows
\begin{align}
   \expHoeffding(\{p\},\{q\}) &= \min_{u:D(u\|q) \le r} D(u\|p) \nonumber \\
   &= \sup_{0 \le \lambda \le 1} \frac{-\lambda r -\psi_{\lambda}(p\|q)}{1-\lambda} \label{eqn:hoeff_exp}
\end{align}
The first equality has a geometric interpretation: $u$ is the closest distribution to $p$ (in KL divergence) from a KL divergence ball of radius $r$ centered at $q$. From \cite[Equation 7,8]{hayashi2009discrimination}, we know that there exists a unique $\lambda$ which attains the optimum in \eqref{eqn:hoeff_exp}. Let
\begin{align} \label{eqn:hoeff_pair}
    (\pHStar, \qHStar)  &:=  \argmin \limits_{p \in \mathcal{P}, q \in \mathcal{Q}} \ \sup_{0 \le \lambda \le 1} \frac{-\lambda r -\psi_{\lambda}(p\|q)}{1-\lambda}
\end{align}
$(\pHStar, \qHStar)$ can be interpreted as the hardest pair of distributions for the problem of characterizing optimal tradeoff for adversarial hypothesis testing. Let $s^*$ be the Hoeffding exponent for $(\pHStar, \qHStar)$ and $\lambda^*$ be the corresponding unique optimizer.
\begin{align}
    \lambda^* &:= \arg\sup \limits_{0 \le \lambda \le 1} \frac{-\lambda r -\psi_{\lambda}(\pHStar\|\qHStar)}{1-\lambda} \label{eqn:lambda_star} \\
    s^* &:= \frac{-\lambda^* r -\psi_{\lambda^*}(\pHStar\|\qHStar)}{1-\lambda^*} \label{eqn:hoeff_exp}
\end{align}
We now prove the following lemma which will help us to show that the likelihood ratio test w.r.t. the pair $(\pHStar,\qHStar)$ obtains the optimal tradeoff.
\begin{lemma}
For any distribution $q \in \mathcal{Q}$, we have
    \begin{align} \label{eqn:lemma_ineq1}
        \sum_{x \in \Omega} q(x) \inp{\frac{\pHStar(x)}{\qHStar(x)}}^{1-{\lambda}^*} \le \sum_{x \in \Omega} \qHStar(x) \inp{\frac{\pHStar(x)}{\qHStar(x)}}^{1-{\lambda}^*} = 2^{\psi_{\lambda}(\pHStar\|\qHStar)}
    \end{align}
For any distribution $p \in \mathcal{P}$, we have    
    \begin{align} \label{eqn:lemma_ineq2}
        \sum_{x \in \Omega} p(x) \inp{\frac{\qHStar(x)}{\pHStar(x)}}^{{\lambda}^*} \le \sum_{x \in \Omega} \pHStar(x) \inp{\frac{\qHStar(x)}{\pHStar(x)}}^{{\lambda}^*} = 2^{\psi_{\lambda}(\pHStar\|\qHStar)}
    \end{align}
\end{lemma}
\begin{proof}
    We first show \eqref{eqn:lemma_ineq1}. Consider the function
    \begin{align*}
        F_{p,q}(\lambda) = \frac{-\lambda r -\psi_{\lambda}(p\|q)}{1-\lambda} 
    \end{align*}
    For $t \in [0,1]$ and $q \in \mathcal{Q}$, let $q_t = tq + (1-t)\qHStar$. Let $f:[0,1]^2 \rightarrow \mathbb{R}$ be defined by
    \begin{equation*}
        f(\lambda,t) := F_{\pHStar,q_t}(\lambda)
    \end{equation*}
    Define $V(t) := \sup_{\lambda \in (0,1)} f(\lambda,t)$. As we have shown, for any $t \in [0,1]$, there exists a unique $\lambda^*(t) \in (0,1)$ such that $V(t) = f(\lambda^*(t),t)$. By chain rule, we have
    \begin{align*}
        V'(t) &= \frac{\partial f(\lambda^*(t),t)}{\partial \lambda^*(t)} . \frac{\partial \lambda^*(t)}{\partial t}  + \frac{\partial f(\lambda^*(t),t)}{\partial t} \\
        &= \frac{\partial f(\lambda^*(t),t)}{\partial t}
    \end{align*}
    The second equality holds since $\frac{\partial f(\lambda^*(t),t)}{\partial \lambda^*(t)} = 0$ by optimality of $\lambda^*(t)$. We now compute $V'(t)$.
    \begin{align*}
        V'(t) &= \frac{\partial f(\lambda^*(t),t)}{\partial t} \\
        &= \frac{-1}{1-\lambda^*(t)} \sum_{x} \frac{ \lambda^*(t) (tq+(1-t)\qHStar)^{\lambda^*(t)-1}(q-\qHStar) (\pHStar)^{1-\lambda^*(t)} }{2^{\psi_{\lambda^*(t)}}}
    \end{align*}
    For $t=0$, we have $q_0=\qHStar$ and $\lambda^*(0)=\lambda^*$ by \eqref{eqn:lambda_star}. The optimality of $\qHStar$ as defined in \eqref{eqn:hoeff_pair} implies $V'(0) \ge 0$. Substituting $t=0$ in the above equation, we get
    \begin{align*}
        0 \le V'(0) &= \frac{-1}{1-\lambda^*} \sum_{x} \lambda^* (\qHStar)^{\lambda^*-1}(q-\qHStar)(\pHStar)^{1-\lambda^*}
    \end{align*}
    Rearranging the terms, we get 
    \begin{align*}
        0 &\le \sum_{x} (q-\qHStar) \left( \frac{\pHStar}{\qHStar} \right)^{1-\lambda^*} \\
        \sum_{x \in \Omega} q(x) \inp{\frac{\pHStar(x)}{\qHStar(x)}}^{1-{\lambda}^*} &\le \sum_{x \in \Omega} \qHStar(x) \inp{\frac{\pHStar(x)}{\qHStar(x)}}^{1-{\lambda}^*} = 2^{\psi_{\lambda}(\pHStar\|\qHStar)}.
    \end{align*}
    To show \eqref{eqn:lemma_ineq2}, we let $p_t = tp + (1-t)\pHStar$ for some $t \in [0,1], p \in \mathcal{P}$ and define $f:[0,1]^2 \rightarrow \mathbb{R}$ to be
    \begin{equation*}
        f(\lambda,t) := F_{p_t,\qHStar}(\lambda)
    \end{equation*}
    The rest of the proof is similar to the proof of \eqref{eqn:lemma_ineq1}.
\end{proof}

\begin{theorem} \label{thrm:bemv4}
    \begin{equation}
        \expHoeffding(\mathcal{P},\mathcal{Q}) = \sup_{0 \le \lambda \le 1} \frac{-\lambda r -\psi_{\lambda}(\pHStar\|\qHStar)}{1-\lambda} = s^*
    \end{equation}
\end{theorem}
\begin{proof}
    \subsection{Achievability}
    Consider the test with the following decision rule.
    \begin{align*}
        Z_n &= \begin{cases}
            1 &\text{ if } \frac{\pHStar(X_i)}{\qHStar(X_i)} > 2^{n(r-s^*)} \\
            0 &\text{ otherwise }
        \end{cases}
    \end{align*}
    We need to show that: \\
    (i) Under $H_1$, for all possible adversary strategies $\mathcal{A}_n$, we have $\pi_{0|1}(\phi_n, \mathcal{A}_n) \le 2^{-nr}$. \\
    (ii) Under $H_0$, for all possible adversary strategies $\mathcal{A}_n$, we have $\pi_{1|0}(\phi_n, \mathcal{A}_n) \le 2^{-ns^*}$. \\
    We first show $(i)$.
    \begin{align*}
        \pi_{0|1}(\phi_n, \mathcal{A}_n) &= \mathbb{Q} \inp{ \prod_{i=1}^{n} \frac{\pHStar(X_i)}{\qHStar(X_i)} > 2^{n(r-s^*)}} \\
        &= \mathbb{Q} \inp{ \prod_{i=1}^{n} \frac{\pHStar(X_i)^{1-\lambda^*}}{\qHStar(X_i)^{1-\lambda^*}} > 2^{n({1-\lambda^*})(r-s^*)} } \\
        &\overset{(a)}{\le} 2^{-n({1-\lambda^*})(r-s^*)} \mathbb{E}_{\mathbb{Q}} \insq{\prod_{i=1}^{n}  \frac{\pHStar(X_i)^{1-\lambda^*}}{\qHStar(X_i)^{1-\lambda^*}}} \\
        &= 2^{-n({1-\lambda^*})(r-s^*)} \mathbb{E} \insq{ \mathbb{E} \insq{\prod_{i=1}^{n} \frac{\pHStar(X_i)^{1-\lambda^*}}{\qHStar(X_i)^{1-\lambda^*}} \bigg| X_1,\ldots,X_{n-1} } } \\
        &= 2^{-n({1-\lambda^*})(r-s^*)} \mathbb{E} \insq{ \prod_{i=1}^{n-1} \frac{\pHStar(X_i)^{1-\lambda^*}}{\qHStar(X_i)^{1-\lambda^*}} \mathbb{E} \insq{ \frac{\pHStar(X_n)^{1-\lambda^*}}{\qHStar(X_n)^{1-\lambda^*}} \bigg| X_1,\ldots,X_{n-1} } }
    \end{align*}
    $(a)$ is by Markov's inequality. The next two equalities follow from the law of total expectation and the fact that $\prod_{i=1}^{n-1} \frac{\pHStar(X_i)^{1-\lambda^*}}{\qHStar(X_i)^{1-\lambda^*}}$ is a function of $X_1,\ldots,X_n$. Now, observe the inner expectation. Conditioned on past observations, the adversary will choose a distribution $\hat{q}_{n-1}(X_1,\ldots, X_n)$ from the set $\mathcal{Q}$. We can then use \eqref{eqn:lemma_ineq1} to bound the inner expectation.
    \begin{align*}
        \pi_{0|1}(\phi_n, \mathcal{A}_n) &\le 2^{-n({1-\lambda^*})(r-s^*)} 2^{\psi_{\lambda^*}(\pHStar\|\qHStar)} \mathbb{E} \insq{ \prod_{i=1}^{n-1} \frac{\pHStar(X_i)^{1-\lambda^*}}{\qHStar(X_i)^{1-\lambda^*}}  } \\
        &\overset{(b)}{\le} 2^{-n({1-\lambda^*})(r-s^*)} 2^{n\psi_{\lambda^*}(\pHStar\|\qHStar)} \\
        &= 2^{-nr} 2^{n(\lambda^*r+(1-\lambda^*)s^* + \psi_{\lambda^*}(\pHStar\|\qHStar))} \\
        &\overset{(c)}{=} 2^{-nr}
    \end{align*}
    $(b)$ follows from repeated application of \eqref{eqn:lemma_ineq1}. $(c)$ follows from the definition of $s^*$ in \eqref{eqn:hoeff_exp}. \\
    We now show (ii).
    \begin{align*}
        \pi_{1|0}(\phi_n, \mathcal{A}_n) &= \mathbb{P} \inp{ \prod_{i=1}^{n} \frac{\pHStar(X_i)}{\qHStar(X_i)} \le 2^{n(r-s^*)}} \\
        &= \mathbb{P} \inp{ \prod_{i=1}^{n} \frac{\qHStar(X_i)^{\lambda^*}}{\pHStar(X_i)^{\lambda^*}} > 2^{-n{\lambda^*}(r-s^*)}} \\
        &\overset{(a)}{\le} 2^{n{\lambda^*}(r-s^*)} \mathbb{E}_{\mathbb{P}} \insq{ \prod_{i=1}^{n} \frac{\qHStar(X_i)^{\lambda^*}}{\pHStar(X_i)^{\lambda^*}} } \\
        &= 2^{n{\lambda^*}(r-s^*)}  \mathbb{E} \insq{ \mathbb{E} \insq{\prod_{i=1}^{n} \frac{\qHStar(X_i)^{\lambda^*}}{\pHStar(X_i)^{\lambda^*}} \bigg| X_1,\ldots,X_{n-1} } } \\
        &= 2^{n{\lambda^*}(r-s^*)}  \mathbb{E} \insq{ \prod_{i=1}^{n-1} \frac{\qHStar(X_i)^{\lambda^*}}{\pHStar(X_i)^{\lambda^*}} \mathbb{E} \insq{ \frac{\qHStar(X_n)^{\lambda^*}}{\pHStar(X_n)^{\lambda^*}} \bigg| X_1,\ldots,X_{n-1} } }
    \end{align*}
    $(a)$ is by Markov's inequality. The next two equalities follow from the law of total expectation and the fact that $\prod_{i=1}^{n-1} \frac{\qHStar(X_i)^{1-\lambda^*}}{\pHStar(X_i)^{1-\lambda^*}}$ is a function of $X_1,\ldots,X_n$. Now, observe the inner expectation. Conditioned on past observations, the adversary will choose a distribution $\hat{p}_{n-1}(X_1,\ldots, X_n)$ from the set $\mathcal{P}$. We can then use \eqref{eqn:lemma_ineq2} to bound the inner expectation.
    \begin{align*}
        \pi_{1|0}(\phi_n, \mathcal{A}_n) &\le 2^{n{\lambda^*}(r-s^*)} 2^{\psi_{\lambda^*}(\pHStar\|\qHStar)} \mathbb{E} \insq{ \prod_{i=1}^{n-1} \frac{\qHStar(X_i)^{\lambda^*}}{\pHStar(X_i)^{\lambda^*}} } \\
        &\overset{(b)}{\le} 2^{n{\lambda^*}(r-s^*)} 2^{n\psi_{\lambda^*}(\pHStar\|\qHStar)} \\
        &= 2^{-ns^*} 2^{n(\lambda^*r + (1-\lambda^*)s^* + \psi_{\lambda^*}(\pHStar\|\qHStar) )} \\
        &\overset{(c)}{=} 2^{-ns^*}
    \end{align*}
    $(b)$ follows from repeated application of \eqref{eqn:lemma_ineq2}. $(c)$ follows from the definition of $s^*$ in \eqref{eqn:hoeff_exp}.
    \subsection{Converse}
    Fix the adversary strategy to be i.i.d. $\pHStar$ under $H_0$ and i.i.d. $\qHStar$ under $H_1$. Applying \eqref{eqn:hoeff_exp} finishes the converse.
\end{proof}
%\balance
\newpage
\bibliographystyle{ieeetr}

\bibliography{refs}
%\clearpage

\end{document}